\documentclass[12pt]{amsart}
\usepackage{mathrsfs}
\usepackage{amssymb} %Extra symbols 
\usepackage{amsmath} %Math commands 
\usepackage{url}
\usepackage{color}
\usepackage{hyperref}
\definecolor{refcolor}{RGB}{0,0,190}
\hypersetup{
    colorlinks,
    citecolor=refcolor,
    filecolor=refcolor,
    linkcolor=refcolor,
    urlcolor=refcolor
}
\usepackage{graphicx} %Include postscript graphics 
\usepackage{amsaddr}

\newtheorem{theorem}{Theorem}

\newtheorem{definition}{Definition}

\newtheorem{principle}{Principle}

\newtheorem{condition}{Condition}

\newenvironment{proof-argument}{\begin{proof}[Argument.]}{\end{proof}}
\newenvironment{proof-evidence}{\begin{proof}[Evidence.]}{\end{proof}}
\def\({\left(}
\def\){\right)}

\newcommand{\tn}{\textnormal}

\newcommand{\mc}[1]{\mathcal{#1}}
\newcommand{\ms}[1]{\mathscr{#1}}

\newcommand{\R}{\mathbb{R}}
\newcommand{\C}{\mathbb{C}}
\newcommand{\N}{\mathbb{N}}
\newcommand{\Z}{\mathbb{Z}}
\newcommand{\abs}[1]{\left|#1\right|}

\newcommand{\bra}[1]{\langle#1|}
\newcommand{\ket}[1]{|#1\rangle}

\newcommand{\schrod}{Schr\"odinger}
\newcommand{\pwt}{pilot-wave theory}

\def\hilbert{\mathcal{H}}

\def\sref #1{\S\ref{#1}}

\newcommand{\ie}{\emph{i.e.}\ }
\newcommand{\eg}{\emph{eg.}\ }
\newcommand{\cf}{\emph{cf.}\ }
\newcommand{\etc}{\emph{etc}}
\newcommand{\etal}{\emph{et al.}}

\newcommand{\x}{\mathbf{x}}
\newcommand{\M}{\mathbb{E}_3}
\newcommand{\D}{\mathbf{D}}

\newcommand{\n}{\mathbf{N}}
\newcommand{\V}[1]{\mathbb{V}_{#1}}
\newcommand{\nD}[1]{|{#1}|}

\begin{document}

\author{Ovidiu Cristinel Stoica}
\address{Department of Theoretical Physics, \\National Institute of Physics and Nuclear Engineering -- Horia Hulubei, Bucharest, Romania. Corresponding author.}
\email{cristi.stoica@theory.nipne.ro; holotronix@gmail.com}

\title{The post-determined block universe}

%------------------------------------------------------------%
\begin{abstract}

A series of reasons to take quantum unitary evolution seriously and explain the projection of the state vector as unitary and not discontinuous are presented, including some from General Relativity. This leads to an interpretation of Quantum Mechanics which is unitary at the level of a single world. I argue that unitary evolution is consistent with both quantum measurements and the apparent classicality at the macroscopic level. This allows us to take the wavefunction as ontic (but holistic), but a global consistency condition has to be introduced to ensure this compatibility. I justify this by appealing to sheaf cohomology on the block universe.

As a consequence, Quantum Theory turns out to be consistent with a definite four-dimensional spacetime, even if this may consist of superpositions of different geometries. 

But the block universe subject to global consistency gains a new flavor, which for an observer experiencing the flow of time appears as ``superdeterministic'' or ``retrocausal'', although this does not manifest itself in observations in a way which would allow the violation of causality. However, the block universe view offers another interpretation of this behavior, which makes more sense, and removes the tension with causality.

Such a block universe subject to global consistency appears thus as being post-determined. Here ``post-determined'' means that for an observer the block universe appears as not being completely determined from the beginning, but each new quantum observation eliminates some of the possible block universe solutions consistent with the previous observations.

I compare the post-determined block universe with other proposals: the presentist view, the block universe, the splitting block universe, and the growing block universe, and explain how it combines their major advantages in a qualitatively different picture.
\end{abstract}

%\pacs{03.65.Ta}{Foundations of quantum mechanics; measurement theory}
%\pacs{03.65.-w}{Quantum mechanics}


\maketitle

%------------------------------------------------------------%
\section{Introduction}
\label{s:intro}

The most successful theories we have, Quantum Mechanics (QM) and (classical) General Relativity (GR), seem to be mutually inconsistent, to have some internal tensions, and also some tensions with other observed phenomena. 
Should we, in order to solve these tensions, make radical changes, or should we rather take a more conservative route?
%If we learned something from the history of physics, it is that even the best theories we have tend to be replaced by better theories, so we should expect that this will likely continue to happen. However, at least some important principles survived when better theories replaced the old ones.

The aim of this article is to show that an interpretation of QM \cite{Sto16aWavefunctionCollapse,Sto17UniverseNoCollapse}, which relies on a new kind of block universe, is capable to reduce some of the tensions between QM and GR, in a most conservative way.

It is reasonable to expect some conditions to be satisfied by any theory aiming to conciliate QM with GR.
In the following I will enumerate some of them. Unfortunately, no theory can satisfy all of them simultaneously, and there is no objective or unanimous criterion that would allow us to decide which ones to choose and which ones to abandon.
In this article, I will argue that adopting a conservative position about GR leads to a particular choice, which is also conservative about QM, although not all conditions will be satisfied.

\begin{condition}[Ontology]
\label{condition:ontology}
The theory should have an ontology.
\end{condition}

Classical GR has a well defined ontology. One of the tensions between QM and GR comes from the fact that it is difficult to associate an ontology to QM in a consistent way. Various interpretations propose different ontologies for QM, with different advantages and disadvantages. Therefore, it is important to see which of these proposals are consistent with the ontology of GR, or if GR can be changed to accommodate them.

\begin{condition}[Spacetime+quantum]
\label{condition:spacetimeQ}
The theory's entities should be defined on the relativistic spacetime.
\end{condition}

This is a problem dating since the formulation of wave mechanics by {\schrod}. The problem is that {\schrod}'s equation has as solutions wavefunctions defined on the configuration space, rather than on the physical $3$-dimensional space or on spacetime. This Condition is sometimes conflated with the Condition \ref{condition:ontology}, but they are different, since an interpretation can have an ontology defined on the configuration space. But a configuration space ontology may be in tension with Special and General Relativity, which are based on a four-dimensional spacetime.

Note that in GR the dimension of spacetime is a topological property, and the Lorentzian metric $g$ is a field defined on the four-dimensional topological manifold. It may be the case that when quantum effects are taken into account, $g$ itself becomes quantized as a field, without affecting the topology. There are proposals that the topology itself is affected, or that spacetime is in fact discrete at very low scales, but in this paper I will try to be as conservative as possible.
This leads to the following.
\begin{condition}[Spacetime+metric]
\label{condition:spacetimeg}
While the metric tensor can be treated as a field when quantized, it should remain defined on a four-dimensional topological manifold.
\end{condition}

The fact that repeated measurements of the same quantum system lead to inconsistencies when the corresponding operators do not commute is the reason for postulating the wavefunction collapse or the projection postulate. It is not specified what dynamics can lead to the collapse due to measurements, and this is also part of the measurement problem. 

\begin{condition}[Universality]
\label{condition:universal_law}
The dynamics should be universal, \ie valid at all times.
\end{condition}

The postulated wavefunction collapse or projection prescribes that the {\schrod} equation is suspended when the projection takes place. Having a law suspended and replaced with another at times means that the dynamics is not universal. Various interpretations come with different solutions. Not all of them solve the problem without requiring that the {\schrod} equation is suspended or broken during quantum measurements.

\begin{condition}[Conservation laws]
\label{condition:conservation}
The dynamics should satisfy the conservation laws.
\end{condition}

A perhaps less acknowledged prediction of Standard QM is that the projection postulate leads to violations of the conservation laws \cite{Sto17UniverseNoCollapse}. Conservation laws are fundamental in both the Hamiltonian and the Lagrangian formulations of physical theories. They are related to the symmetries of the theories, and the dynamics itself depends on the symmetries.

\begin{condition}[Relativity of simultaneity]
\label{condition:relativity_of_simultaneity}
The theory should satisfy relativity of simultaneity.
\end{condition}

A quantum theory consistent with Special Relativity needs to satisfy Lorentz and in fact Poincar\'e invariance, and hence relativity of simultaneity. To be consistent with GR, it needs to be invariant to diffeomorphisms, and also to satisfy the relativity of simultaneity, with respect to foliations of spacetime into spacelike hypersusrfaces.

\begin{condition}[Locality]
\label{condition:locality}
The dynamics of the theory should be local.
\end{condition}

It would be interesting to see if there is a way to have a definite ontology in a nonlocal theory, without violating the relativity of simultaneity (\cf Condition \ref{condition:relativity_of_simultaneity}).

\begin{condition}[Strong causality]
\label{condition:strong_causality}
The theory should satisfy strong causality.
\end{condition}
A theory satisfies strong causality if it can be formulated so that the state at any time $t_0$ depends on the states in the past times $t<t_0$ only.
This may seem a trivial demand, but not all interpretations of QM satisfy it, for various justifiable reasons.
Strong causality does not necessarily require determinism, it only means that one does not need to appeal to the future to determine the state or the probability distribution of the states at the time $t_0$.
Since most theories are either time reversible or at least partially reversible, the dynamics allows to infer information about the state at a time $t_0$ from the knowledge of future states at $t>t_0$, but the point of Condition \ref{condition:strong_causality} is that the dynamics of the theory can be formulated in a way that does not need to appeal to the future states.

\begin{condition}[Weak causality]
\label{condition:weak_causality}
The theory should satisfy weak causality.
\end{condition}
A theory satisfies weak causality if it can be formulated so that the coarse grained state or the probability distribution of the states at any time $t_0$ depends on the coarse grained states in the the past times $t<t_0$ only.
By \emph{coarse grained} states it is understood the class of equivalence of states, at a microscopic or fundamental level, which the observer does not distinguish.
Condition \ref{condition:weak_causality} is similar to Condition \ref{condition:strong_causality}, the difference being that Condition \ref{condition:weak_causality} applies to coarse grained states, and Condition \ref{condition:strong_causality} to micro states.

\begin{condition}[Free will]
\label{condition:free_will}
The theory should allow any experimental set-up possible in principle at the macro level.
\end{condition}
This notion of free will is not taken here in the sense from psychology or philosophy of mind, it just refers to the operational requirement that the agents are allowed by the physical laws to arrange the experimental setups for performing quantum measurements in any way possible in principle. For example, that they are not forced to orient a Stern-Gerlach device only along some particular directions. Being about the experimental setup, it refers to the coarse grained level where the world looks classical, and not to the micro or fundamental level.

\begin{condition}[Macro-quasiclassicality]
\label{condition:quasiclassicality}
The theory should be consistent with the observation that the world appears quasiclassical at macro level.
\end{condition}
Our mundane experience is consistent with a classical world. But we know that the world is quantum, and apparently {\schrod}'s equation predicts that one should expect strange superpositions at the macro level too. The fact that we do not see such superpositions needs to be explained by the theory, or at least allowed. The point is not to have full macrorealism (\ie any system should be in only one macro state at a time, and that state can be determined by noninvasive measurements), since this is prevented by the violation of the Leggett-Garg inequality \cite{leggett2002limitsQM,EmaryNeillFranco2013LeggettGarg,LeggettGarg1985QMvsMacrorealism,Leggett2008Realism}, but to allow, and if possible, to explain, how the quantum world can appear to us quasiclassical.
Perhaps a better way to formulate this condition is by requiring the coarse grained level of a quantum world to be like that of a classical world.

The fact that the world appears classical to the observer is also related to the \emph{measurement problem}, since measuring a system in a superposition is expected to lead to a superposition of different experimental outcomes, which include the apparatus being entangled with the possible outcomes as well, and even the observer. But the fact that we do not observe such states is just a consequence of the fact that the world has to appear quasiclassical at macro level. What the measurement problem contains in addition to this is the Born rule.
\begin{condition}[Probabilities]
\label{condition:probabilities}
The theory should predict, or at least be consistent with the probabilities of the outcomes as predicted by the Born rule.
\end{condition}

\begin{condition}[Quantum Einstein equation]
\label{condition:qeinstein}
The theory should allow QFT and GR combine in an Einstein-type equation.
\end{condition}

Perhaps Condition \ref{condition:qeinstein} is a weaker version of the condition of finding a theory of Quantum Gravity.
However, I would keep two conditions separate.
\begin{condition}[Perturbative quantum gravity]
\label{condition:qg}
The theory should allow perturbative quantum gravity.
\end{condition}

In addition, classical GR has singularities,
\begin{condition}[Singularities]
\label{condition:singularities}
The theory should have a solution to the problem of singularities.
\end{condition}

Combined with QM, this leads to the following:
\begin{condition}[Hawking's puzzle]
\label{condition:bh-info}
The theory should allow information to get out of the evaporating black holes.
\end{condition}

Conditions \ref{condition:ontology}, \ref{condition:spacetimeQ}, \ref{condition:spacetimeg}, \ref{condition:universal_law}, \ref{condition:conservation}, \ref{condition:relativity_of_simultaneity}, \ref{condition:locality}, \ref{condition:strong_causality}, \ref{condition:weak_causality}, \ref{condition:qeinstein}, \ref{condition:qg}, and \ref{condition:bh-info} seem more directly related to the tensions between QM and GR, but addressing them should not prevent satisfying conditions \ref{condition:free_will}, \ref{condition:quasiclassicality}, and \ref{condition:probabilities}, which are more specific to QM, and the GR condition \ref{condition:singularities}. However, the latter are interconnected with the former, and play an important role in this proposal.

The number of conditions to be satisfied is quite large and perhaps not exhaustive, but we do not need to satisfy all of them at once in order to remove at least some of the tensions between QM and GR. We will see that the approach proposed here satisfies some of these conditions.

In this article, I apply the conservative route of accepting unitary evolution, and push it as far as possible and see if it really breaks down during measurements, if it is really necessary to invoke a wavefunction collapse or many worlds. This is a natural continuation of some results developed in previous works of the author \cite{Sto08b,Sto12QMb,Sto12QMc,Sto16aWavefunctionCollapse,Sto17UniverseNoCollapse}, and of some earlier ones developed by Schulman \cite{schulman1984definiteMeasurements,schulman1997timeArrowsAndQuantumMeasurement,schulman2012experimentalTestForSpecialState,schulman2016specialStatesDemandForceObserver,Schulman2017ProgramSpecialState} (I will explain the difference between these two approaches in section \sref{s:collapse_problems}).
At the same time, I try to be conservative about GR as well. In the case of GR, the most common proposals are to modify the theory, or to obtain it as a limit of a supposedly better theory, most likely a theory of Quantum Gravity. But, again, before doing this, it would be useful to push the limits of the very principles of GR, and see if and when they break down. There are too many possibilities to take into account when advancing towards Quantum Gravity -- even if most of those that we know are incomplete or have problems -- and the only firm ground we have are the principles that we know to be well tested, so it may help to find out where their limits are, and see if they can agree with each other.

There are strong reasons to expect that the wavefunction is real, ontic, rather than a mere collection of probabilities about unseen states or variables or outcomes \cite{Spekkens2005Contextuality,HarriganSpekkens2010EinsteinIncompleteness,ColbeckRenner2011NoExtensionOfQM,ColbeckRenner2012Reality,PBR2012RealityOfPsi,Hardy2013AreQuantumStatesReal,Ringbauer2015MeasurementsRealityWavefunction,Myrvold2018PsiOntology}, and we will see some arguments supporting this in section \sref{s:reality_beyond_wavefunction}, which is about Condition \ref{condition:ontology}.
Section \sref{s:wavefunction_on_spacetime} deals with Condition \ref{condition:spacetimeQ}, and shows how this problem is addressed in \cite{Sto2019RepresentationOfWavefunctionOn3D}.
In section \sref{s:collapse_problems}, I will explain why, if there is a wavefunction collapse taking place like a discontinuous projection, this leads to violations of Conditions \ref{condition:conservation}  and \ref{condition:relativity_of_simultaneity}.
The possibility that what appears to be collapse takes place by unitary evolution alone is discussed in section \sref{s:no_collapse}. This turns out to be possible in a way consistent with the observations at a macro level (\cf Condition \ref{condition:quasiclassicality}) and to satisfy Conditions \ref{condition:locality} and \ref{condition:relativity_of_simultaneity}, but it can only happen for a small subset of the Hilbert space. This limitation of the allowed solutions gives the appearance of a conspiracy among the initial conditions of the quantum particles, or of retrocausality, depending on how you describe the order of the events in time. Retrocausality seems to prevent Condition \ref{condition:strong_causality} to be satisfied, but is consistent with Condition \ref{condition:weak_causality}. Superdeterminism appears to contradict Condition \ref{condition:free_will}. Surprisingly, even if it seems in the same category as superdeterministic proposals \cite{Sto12QMb}, the post-determined block universe proposal satisfies Conditions \ref{condition:strong_causality}, \ref{condition:weak_causality}, and \ref{condition:free_will}.

Retrocausal interpretations have been known for a long time to be able to save locality \cite{deBeauregard1977TimeSymmetryEinsteinParadox,SEP-2019qm-retrocausality,Rietdijk1978retroactiveInfluence,cramer1986transactional,cramer1988overview,Kastner2012TransactionalBook,Wharton2007TimeSymmetricQM,price2008toyRetrocausality,sutherland2008causallySymmetricBohm,Argaman2008Retrocausality,price2015disentangling,sutherland2017RetrocausalityHelps,EmilyAdlam2018SpookyActionTemporalDistance,aharonov1964time,aharonov1991complete,cohen2016quantum2classical,CohenCortesElitzurSmolin2019RealismAndCausalityI,WhartonArgaman2019BellThereomAndSpacetimeBasedQM}. Bell's theorem forces us to choose between nonlocality, which is at least ontologically at odds with Special Relativity (SR) even if satisfies the no-signaling principle, and retrocausality or other ways to violate statistical independence, which are or can be consistent with relativistic local causality, do not signal backwards in time, but require the past to depend on the future as well as the other way around.
Can there be a principle which, rather than modifying the dynamics or extending the theory with new variables, just keeps the ``good'' solutions of the Schr\"odinger equation, and exhibits this apparent retrocausality in a natural way? In section \sref{s:global_consistency} I argue that this can be achieved by a global consistency condition, without requiring modifications of the important principles, and in fact it can save them. Not only is this not in tension with Special or General Relativity, but in fact the latter can provide an explanation, or at least a framework, for the former. While it will be discussed how probabilities can arise, unfortunately a derivation of the Born rule in this framework is still an open problem. Some apparent tensions between QM and GR are discussed in section \sref{s:quantum_and_relativity}, and it is shown that in this framework some of them are naturally resolved or avoided.

Along this discussion, a version of the block universe emerges as the natural framework for this interpretation of QM (section \sref{s:post_determined_block_universe}). The observer experiencing the flow of time will have an alternative to the retrocausal description, in which the state of the universe is initially undetermined, and becomes more and more determined with each new measurement. From the bird's eye view of an observer outside of time, this appears merely as the result of the global consistency condition.
I call this framework the \emph{post-determined block universe} because of the necessity to take into account all of the constraints imposed by quantum measurements, before determining the global solution.

After the integrative summary and the discussion of the post-determined block universe given in Section \sref{s:post_determined_block_universe}, the remaining open problems are described in Section \sref{s:open}. In Section \sref{s:experiment} it is shown that the proposed interpretation is falsifiable through experiments.

%------------------------------------------------------------%
\section{Is there a reality beyond the wavefunction?}
\label{s:reality_beyond_wavefunction}

Schr\"odinger's equation was discovered when trying to explain the energy levels and structure of the atoms \cite{Sch26}. It was then extended to the relativistic case, including creation and annihilation of particles. The resulting Quantum Field Theory (QFT) provides an accurate description of the behavior of particles.

In the following I will refer to the evolution equation by the name of Schr\"odinger, even though the relativistic versions are due to Klein-Gordon, Dirac, and others, and even if field quantization is in place. I will do this for simplicity, based on the fact that these equations have the general form of a Schr\"odinger equation or the square of such an equation, and the evolution is unitary, of the form \eqref{eq:unitary_evolution}. % (see Fig. \ref{qm_u_evolution.pdf}).

\begin{equation}
\label{eq:unitary_evolution}
\ket{\psi(t)}=\hat U(t,t_0)\ket{\psi(t_0)},
\end{equation}
where $\ket{\psi(t)}$ is an operator defined on a suitable Hilbert space $\hilbert$, and $\hat U$ is the unitary evolution operator.
In QFT the state vectors $\ket{\psi(t)}$ are replaced by linear operators $\hat{\psi}$ acting on a special vacuum state $\ket{0}$ to create the states $\ket{\psi(t)}$, and their evolution is still unitary.

If this would be all there is to be explained, Quantum Theory would be the perfect theory. A simple story that explains almost everything: many-particle states, which evolve in time by being transformed by a unitary operator. Except for the fact that neither the wavefunction nor the quantum fields appear to be objects defined on the physical $3$-dimensional space, which will be handled in Section \sref{s:wavefunction_on_spacetime}, one can consider that $\ket{\psi}$ itself is \emph{ontic}, and maybe this could be all that is needed for the ontology. But this is not the full story.

The world appears macroscopically classical, and measurements have definite outcomes. A quantum observation always finds the observed system to be in an eigenstate of the Hermitian operator corresponding to the observed property, and the outcome to be an eigenvalue. Since the probability that the state was already an eigenstate by chance is zero, there seems to be an inconsistency between the unitary evolution and the fact that the observed system is always found to be in an eigenstate of the Hermitian operator.
To solve this inconsistency, it is proposed that the measurement itself is accompanied by some projection of the state vector which does not seem to follow from the Schr\"odinger equation itself, and even breaks it. This solution allows us to obtain a prediction of the probabilities associated to each outcome of the measurement, but we still need to understand how or why it happens, so we have what is called ``the measurement problem''. These problems are indicators that there's more to the story than it seems.

Bohr's solution was to take as given the classicality of the macroscopic level \cite{Bohr58}, and he described quantum measurements by assuming the apparatus and the outcomes of the measurements to be classical, and by accepting the wave-particle duality. Heisenberg's views cross-pollinated his ideas \cite{Hei58}, leading to the Copenhagen Interpretation, which avoids discussing anything but the outcomes of the measurements. It is often considered that restraining the discussion to the outcomes of the measurements makes the problems vanish, or even that these are not real problems. But not everyone is convinced that this makes the problems go away, as we know from the objections by some, including Einstein and de Broglie. Also Schr\"odinger explained the problem of classicality in his famous cat thought experiment \cite{schrodinger1935SchrodingerCat}.

Avoiding to discuss a problem is useful in solving other independent problems, but it does not make it go away. These foundational problems are important in particular if we really want to understand our world, particularly how General Relativity and Quantum Theory can coexist in a consistent way. 

The fact that they are indeed problems becomes apparent when we try to understand what the wavefunction is. When we are talking about atoms or other systems of particles which are stationary or even interact, the wavefunction seems to be like a classical field if the state is separable, or at least like a classical field on the configuration space in general. 
Something has to be real there, something has to carry the energy and momentum, and curve spacetime according to the Einstein equation. The classical limit of quantum electrodynamics comes with the appropriate stress-energy tensor, which seems to be distributed in space like the wavefunction is. The inertia properties of matter, even when they are not subject to quantum measurement, are consistent again with such a stress-energy tensor. 

But when we perform a measurement, the wavefunction seems to turn into a probabilistic device, whose role is to predict the probabilities to jump from one state to another during a measurement. The probability is given by the squared scalar product between the state before and the eigenstate after the measurement (which is in fact the \emph{Born rule} \cite{Born26} extended from positions to any observable).

Is this a simple ambiguity, or a contradiction? How can we interpret the wavefunction as being ontic when we do not look at it, and epistemic during measurements? Obviously the Copenhagen Interpretation chooses a quick way out, by denying the reality of the wavefunction (or at best claiming that it is irrelevant), to avoid this contradiction.

This apparent contradiction is visible in the more rigorous mathematical formulations of Quantum Mechanics by Dirac \cite{Dirac58PPQM} and von Neumann \cite{vonNeumann1955MathFoundationsQM}. Accordingly, the state vector is always well defined and has a unitary time evolution, according to the Schr\"odinger equation, but when it is measured, it projects to an eigenstate of the operator corresponding to the observable, according to the Born rule. The apparent conflict becomes manifest when we try to formulate the theory in a mathematically rigorous way, and the infamous \emph{wavefunction collapse} seems to be unavoidable.

As the father of General Relativity, Einstein realized this tension between realism and the purely epistemic view of the Copenhagen Interpretation in its full depth. He never ceased to hope that there is a better explanation, and considered that Quantum Mechanics is incomplete in some sense. He made his most concrete formulation of the problem together with his collaborators Podolsky and Rosen in \cite{EPR35}, where they proposed the famous \emph{EPR experiment}.

There are attempts to regard particles as well-localized to a point, and to interpret the wavefunction as just giving the probability to find the point-particle at given positions. But even the most successful theory that postulates point-particles, \emph{de Broglie-Bohm theory} \cite{db1956double_solution,db1956double_solution_brief,Bohm52,DGZ1996BohmianMechanics}, attributes an ontological status to the wavefunction. Regardless of the proposed interpretation of particles as points, to avoid inconsistencies with the experiments, we have to assign physical properties like charge and mass densities, and ultimately all physical properties except for definite positions, to the wave itself. Both de Broglie and Bohm \cite{Bohm52,Bohm95,Bohm2004CausalityChanceModernPhysics}, as well as another major supporter of this interpretation, Bell ~\footnote{According to Bell, \emph{``No one can understand this theory until he is willing to think of [the wavefunction] as a real objective field rather than just a `probability amplitude'. Even though it propagates not in 3-space but in 3N-space.''} \cite{Bell2004SpeakableUnspeakable} p. 128.}, realized this, and kept both the point particle and the wavefunction as ontic.

There are other reasons to take the wavefunction as ontic, rather than as a purely epistemic or probabilistic device, as implied by some results and no-go theorems \cite{Spekkens2005Contextuality,HarriganSpekkens2010EinsteinIncompleteness,ColbeckRenner2011NoExtensionOfQM,ColbeckRenner2012Reality,PBR2012RealityOfPsi,Hardy2013AreQuantumStatesReal,Ringbauer2015MeasurementsRealityWavefunction,Myrvold2018PsiOntology}. To summarize, these results show one thing: under reasonable assumptions, the wavefunction is not a mere predictor of the probabilities, it does not have redundant parts, it does not contain useless information. In other words, even if it will be replaced in an interpretation by something more concrete or more appropriate for one purpose or another, the structures that replaces it will have to contain the complete information contained in the original wavefunction. In other words, whatever ontology will be proposed for QM, it will have to include the wavefunction itself.

This does not mean that the wavefunction does not have an epistemic or probabilistic role to play. It has to have, to account for the quantum probabilities. These are not mutually exclusive, as we shall see in Section \sref{s:global_consistency}. For the moment, let us see how the wavefunction can be seen as an object on the $3$D-space or the $4$D spacetime.

%------------------------------------------------------------%
\section{The wavefunction on spacetime}
\label{s:wavefunction_on_spacetime}

The nonrelativistic wavefunction for $\n$ particles is defined on the configuration $3\n$-dimensional space rather than the physical $3$-dimensional space, and the quantum field is defined as an operator on an even more abstract and higher-dimensional space. This problem bothered from the dawn of QM physicists like {\schrod}, Lorentz, Heisenberg, Einstein and others \cite{Przibram1967LettersWaveMechanics,Przibram2011LettersWaveMechanics,Schrodinger1926QuantisierungAlsEigenwertproblem,Schrodinger2003CollectedPapersWaveMechanics,Howard1990EinsteinWorriesQM,FineBrown1988ShakyGameEinsteinRealismQT,Bohm2004CausalityChanceModernPhysics}.

In \cite{Sto2019RepresentationOfWavefunctionOn3D} a representation of the wavefunction on the $3$D-space $\M$ was constructed. The representation is in terms of classical fields on the $3$D-space. As long as the evolution is unitary, \ie without collapse, following the {\schrod} equation, all interactions and the propagation are local, satisfying Condition \ref{condition:locality}.

The construction is made in multiple steps. First, separable states of the form
\begin{equation}
\Psi(\x_1,\ldots,\x_\n)=\psi_1(\x_1)\ldots\psi_\n(\x_\n)
\end{equation}
can be seen, up to a phase factor, as $\n$ wave functions defined on the $3$D-space. If each of the functions $\psi_j(\x_j)$, $j\in\{1,\ldots\n\}$, is represented as a section of a vector bundle $\V{\M,\D_j}$, where the vector space which is the fiber represents the spin and the internal degrees of freedom, labeled with elements from a set $\D_j$, then we can represent the $\n$ functions $\psi_1(\x_1),\ldots,\psi_\n(\x_\n)$ on the direct sum of these vector bundles, $\V{\M,\D_1}\oplus\ldots\oplus\V{\M,\D_\n}$. 
But this representation is redundant, because of the phase factors. To eliminate the redundancy, instead of the sum of the $\n$ vector bundles, a quotient is made by an equivalence relation $\sim$, defined for $\n=1$ as $(\psi)\sim(\psi')$ iff $\psi=\psi'$.
For $\n>1$, it is defined by
\begin{equation}
\label{eq:btimes_equiv_n}
(\psi_1,\ldots,\psi_\n)\sim(\psi'_1,\ldots,\psi'_2)
\end{equation}
iff there are $\n$ complex numbers $c_1,\ldots,c_\n$, so that $c_1\ldots c_\n=1$, and the linear transformation of $\V{\M,\D_1}\oplus\ldots\oplus\V{\M,\D_\n}$ dewfined as
\begin{equation}
\label{eq:btimes_transf_n}
T=
\(
 \begin{matrix}
  c_1 1_{\nD{\D_1}} & 0 & \ldots & 0\\
	\ldots & \ldots & \ldots & \ldots \\
  0 & \ldots & 0 & c_\n 1_{\nD{\D_\n}}
 \end{matrix}
\),
\end{equation}
satisfies $(\psi'_1,\ldots,\psi'_2) = T\(\psi_1,\ldots,\psi_\n\)$.

We see that even for separable states the construction is not straightforward. For nonseparable states, one can imagine that we can sum the elements of the bundle obtained by factoring out the equivalence relation $\sim$. But this would not work, instead, we need to take a basis of the Hilbert space of each type of particle, make the construction for the products of the elements of these bases, and only then take the direct sum of the resulting bundles.
This is shown in \cite{Sto2019RepresentationOfWavefunctionOn3D} to be a faithful representation of the tensor products of one-particle Hilbert spaces, as fields on $\M$.

But this is still not enough, since the construction had to be shown to satisfy \emph{local separability}, \ie it should allow the definition of fields on subsets  from $\M$, so that if the fields defined locally on two subsets $A,B\in\M$ are equal on the intersection $A\cap B$, they can be extended as a field on the union $A\cup B$. In \cite{Sto2019RepresentationOfWavefunctionOn3D}, this was done by promoting the global symmetry defined by transformations of the form \ref{eq:btimes_transf_n} to a local symmetry. In other words, local separability was obtained as a gauge theory on the bundle construction used to represent the wavefunction as fields on $\M$.

Then, the propagation of the fields representing free particles was shown to take place locally on $\M$. If the potentials also propagate locally, the entire {\schrod} was represented as a local field equation on $\M$. All interactions and propagations are local as long as they are so in the classical theory that was quantized. Locality holds as long as there is no collapse.

The construction was applied to represent states in QFT, again, in terms of classical fields on $\M$.

In the article \cite{Sto2019RepresentationOfWavefunctionOn3D} it is acknowledged that this representation is artificial and has the purpose to prove the possibility of such representations, the proof of concept. Perhaps simpler and more natural ways to do this will be found in the future, but what is important to keep in mind is that it is not the case that QM or QFT require with necessity that the objects cannot be defined on the $3$D-space $\M$ or on spacetime. In fact, the construction, even as complicated as it may seem, allows the physicist to continue to use the same tensor product formalism as usually, and at the same time to conceptualize the wavefunctions or fields as existing in multiple layered fields on $\M$ or on spacetime.

For the reasons mentioned in Section \sref{s:reality_beyond_wavefunction}, I suggest that if we insist on assigning an ontological interpretation to quantum particles, this should be done by considering the wavefunctions or quantum fields as fields, similar to the classical ones. Now we see that, while the wavefunction is normally defined on the configuration space and not simply on the physical space, this is just a particularity of the representation, and a representation on the $3$D-space is possible. Moreover, these fields behave locally under the unitary evolution.
Therefore, assigning an ontological interpretation to the wavefunction as object in the physical space is possible, as long as no collapse is involved.
We will see in the following that the wavefunction collapse introduces other problems.

%------------------------------------------------------------%
\section{What would happen if the wavefunction collapses would be discontinuous?}
\label{s:collapse_problems}

Although it is in general considered a settled issue that there is a collapse of the wavefunction, or a projection postulate (in the case of the MWI there is an effective collapse for the single worlds), I would like to reopen this case and challenge it. Here are some reasons which I think make the case that it worth considering the possibility that single worlds may evolve unitarily, without collapse, and still allow quantum measurements to have definite outcomes as we observe.
This will not mean to leave unexplained the probabilistic nature of quantum measurements, and I will discuss this in the subsequent sections. For the moment, I would just want to list some of the arguments to consider the possibility that the collapse does not take place in a discontinuous way which violates the unitary evolution.

A first reason is that without collapse, Condition \ref{condition:ontology}, of ontology, would be easily satisfied. We just take the wavefunction as being the real entity in any quantum theory. Moreover, as explain in Section \sref{s:wavefunction_on_spacetime} and \cite{Sto2019RepresentationOfWavefunctionOn3D}, this ontology can be interpreted as being on the physical $3$D-space, \cf Condition \ref{condition:spacetimeQ}.

A condition we used to expect to be satisfied by any fundamental theory is that of universality (Condition \ref{condition:universal_law}), that the law is never broken. This changed with the advent of QM, especially since the dominant view was the Copenhagen Interpretation since the beginning. It is true that this is a subjective matter, since it depends on the interpretation preferred by the physicist or philosopher of physics, but nevertheless, it would be of interest at least to some of them to have a dynamical law that remains true all the time. In fact, this is achieved in MWI.

Another reason is that only unitary evolution can ensure the conservation laws (Condition \ref{condition:conservation}). A discontinuous collapse would break them \cite{Sto17UniverseNoCollapse}. I will come back to this in the remaining part of this section with evidence that collapse would lead to such violations. This may again be a matter of taste, but it is a scientific bet, since it is possible at least in principle, to test whether or not the conservation laws are broken by quantum measurements in our universe (see Section \sref{s:experiment}). Again, by being unitary, MWI is able to ensure the conservation laws, but not for the branches, and this has the same experimental consequences as an actual collapse.

The possibility to explain the outcomes of measurements without collapse would also ensure that Conditions \ref{condition:relativity_of_simultaneity} (relativity of simultaneity) and \ref{condition:locality} (locality) are satisfied (\cf Section \sref{s:wavefunction_on_spacetime}).

In this Section, I will discuss in more detail how collapse violates conservation laws (Condition \ref{condition:conservation}), and why unitary evolution ensures them.

We could guess already that as long as we assume that the wavefunction collapses, the conservation laws are violated, because in Quantum Mechanics the conserved quantities are those whose operators commute with the Hamiltonian, but during a collapse the Hamiltonian evolution is replaced by a projection, which does not commute with the operators corresponding to conserved quantities. Such violations are shown to accompany the collapse explicitly in \cite{Burgos1993ConservationLawsQM,Sto17UniverseNoCollapse}.
This violation happens even if we include the apparatus or the environment, if we assume the \emph{standard measurement scheme}.

Let $\hat{\mc{O}}$ be the operator corresponding to the measured property $\mc{O}$. Suppose $\hat{\mc{O}}$ has a complete set of unit eigenvectors $\ket{j}$, and let $\lambda_j$ be its eigenvalues. Let $\ket{\psi}$ is the state whose property $\mc{O}$ will be measured. Let us consider now that the apparatus also includes the environment, so that the whole universe consists of the observed system in the state $\ket{\psi}$, and the apparatus ready for the experiment, in the state $\ket{\tn{ready}}$.
According to the standard measurement scheme, the apparatus is such that, if the observed system is in an eigenstate $\ket{j}$ of $\hat{\mc{O}}$, the evolution governed by the {\schrod}'s equation leaves the observed system unchanged, and the apparatus ends out in a state $\ket{\tn{outcome }j}$, 
\begin{equation}
\label{eq:outcomej}
\ket{j} \ket{\tn{ready}} \stackrel{\tn{U}}{\longrightarrow} \ket{j} \ket{\tn{outcome }j}.
\end{equation}
The states $\ket{\tn{outcome }j}$ are considered mutually orthogonal for $j\neq j'$, and orthogonal on the state $\ket{\tn{ready}}$.

Following \cite{Burgos1993ConservationLawsQM}, consider that $\hat{A}$ corresponds to a conserved quantity. For simplicity, let's assume that this is an additive quantum quantity, because for multiplicative quantities the proof is similar. Let $\hat{A}^{'}$ and $\hat{A}^{''}$  be the corresponding operators for the observed system, respectively for the system consisting of the apparatus (including the environment). Then, the total conserved operator $\hat{A}$ satisfies
\begin{equation}
\hat{A}=\hat{A}^{'}\otimes 1^{'} + 1^{''}\otimes \hat{A}^{''},
\end{equation}
where $1^{'}$ and $1^{''}$ are the identity operators on the observed system, respectively the apparatus. 

The conserved quantities are those corresponding to the total operator $\hat{A}$. This implies that, under the unitary evolution, its \emph{mean value} for the total system $\ket{\Psi}$,
\begin{equation}
\langle\hat{A}\rangle_{\ket{\Psi}} := \bra{\Psi}\hat{A}\ket{\Psi},
\end{equation}
is conserved. Therefore, considering Eq. \eqref{eq:outcomej}, the mean value of $\hat{A}$ for the full system satisfies
\begin{equation}
\langle\hat{A}\rangle_{\ket{j} \ket{\tn{ready}}} = \langle\hat{A}\rangle_{\ket{j} \ket{\tn{outcome }j}}.
\end{equation}
But since both states are product states of the observed system and the apparatus, from Eq. \eqref{eq:outcomej} it follows that
\begin{equation}
\label{eq:mean_value_additive}
\langle\hat{A}^{'}\rangle_{\ket{j}} + \langle\hat{A}^{''}\rangle_{\ket{\tn{ready}}} = \langle\hat{A}^{'}\rangle_{\ket{j}} + \langle\hat{A}^{''}\rangle_{\ket{\tn{outcome }j}},
\end{equation}
from which we get
\begin{equation}
\label{eq:mean_value_measurement}
\langle\hat{A}^{''}\rangle_{\ket{\tn{ready}}} = \langle\hat{A}^{''}\rangle_{\ket{\tn{outcome }j}}.
\end{equation}

According to the unitary evolution, a general state of the form
\begin{equation}
\label{eq:general_psi}
\ket{\psi}=\sum_j \psi_j\ket{j}
\end{equation}
should evolve into the superposition
\begin{equation}
\label{eq:outcome_super}
\ket{\psi}\ket{\tn{ready}} = \sum_j \psi_j\ket{j} \ket{\tn{ready}} \stackrel{\tn{U}}{\longrightarrow} \sum_j \psi_j \ket{j} \ket{\tn{outcome }j},
\end{equation}
but according to the standard measurement scheme, this superposition should collapse into a state corresponding to a particular $k$, so
\begin{equation}
\label{eq:outcome_super_proj}
\ket{\psi}\ket{\tn{ready}} = \sum_j \psi_j\ket{j} \ket{\tn{ready}} \stackrel{\tn{Pr}}{\longrightarrow} \ket{k} \ket{\tn{outcome }k}.
\end{equation}

Now let us check the conservation of $A$, by calculating the mean value for Eq. \eqref{eq:outcome_super_proj}, using Eq. \eqref{eq:mean_value_additive}
\begin{equation}
\label{eq:outcome_super_proj_mean}
\langle\hat{A}\rangle_{\ket{\psi}\ket{\tn{ready}}} = \langle\hat{A}^{'}\rangle_{\ket{\psi}} + \langle\hat{A}^{''}\rangle_{\ket{\tn{ready}}} \stackrel{\tn{Pr}}{\longrightarrow} \langle\hat{A}^{'}\rangle_{\ket{k}} + \langle\hat{A}^{''}\rangle_{\ket{\tn{outcome }k}},
\end{equation}
and from Eq. \eqref{eq:mean_value_measurement},
\begin{equation}
\langle\hat{A}^{'}\rangle_{\ket{\psi}} + \langle\hat{A}^{''}\rangle_{\ket{\tn{ready}}} \stackrel{\tn{Pr}}{\longrightarrow} \langle\hat{A}^{'}\rangle_{\ket{k}} + \langle\hat{A}^{''}\rangle_{\ket{\tn{ready}}}.
\end{equation}
Therefore, the quantity $A$ is conserved after the measurement iff
\begin{equation}
\label{eq:condition_conservation_collapse}
\langle\hat{A}^{'}\rangle_{\ket{\psi}} = \langle\hat{A}^{'}\rangle_{\ket{k}}.
\end{equation}
This condition is very restrictive.
The possible values of $\langle\hat{A}^{'}\rangle_{\ket{\psi}}$ can be any value in the minimal interval containing the spectrum of $\hat{A}^{'}$, so they form a continuum. But identity \eqref{eq:condition_conservation_collapse} restricts this continuum to only one value, $\langle\hat{A}^{'}\rangle_{\ket{k}}$, which implies that the initial state $\ket{\psi}$ has to be constrained to a zero-measure subset of its Hilbert space.

In particular, if the measured property is $A'$, \ie $\hat{\mc{O}}=\hat{A}^{'}$, then, from Eq. \eqref{eq:general_psi},
\begin{equation}
\label{eq:additive}
\hat{A}^{'}\ket{\psi}=\sum_j \lambda_j \psi_j \ket{j},
\end{equation}
which implies that the mean value of $\hat{A}^{'}$ in the state $\ket{\psi}$ is
\begin{equation}
\label{eq:mean_value}
\langle\hat{A}^{'}\rangle_{\psi}:=\bra{\psi}\hat{A}^{'}\ket{\psi} = \sum_j \abs{\psi_j}^2 \lambda_j.
\end{equation}
Therefore, if $\hat{\mc{O}}=\hat{A}^{'}$, the identity \eqref{eq:condition_conservation_collapse} holds only if $\ket{\psi}=\ket{k}$, which means no collapse took place.

From these considerations it follows that we proved the following:
\begin{theorem}
\label{thm:collapse_breaks_conservation}
The standard measurement scheme implies that the wavefunction collapse violates conservation laws.
\end{theorem}

This result was proven in \cite{Sto17UniverseNoCollapse} by using successive spin measurements along different axes, and obtaining an observable violation of the conservation of the angular momentum. The proof given here has some advantages over the one given in \cite{Sto17UniverseNoCollapse}, since it already shows that the conservation laws are broken from a single measurement, and it is more general.

One may object that if we take into account the environment, then the conservation laws are not broken, since the missing difference may be hidden in the environment.
If this is true, then we have to see what assumptions of the proof are wrong.
The collapse in the standard measurement scheme is the only assumption that conflicts with the conservation laws.
If someday a better description of the measurement process will emerge, it should contradict the standard measurement scheme, by requiring that the states of the form $\ket{\tn{outcome }j}$ depend on the initial state $\ket{\psi}$ of the observed system. This requirement contradicts either the projection postulate, or the principle of superposition. The reason is that, if the projection postulate is maintained, then the principle of superposition implies that $\ket{\tn{outcome }j}$ should be exactly the same as if $\ket{\psi}=\ket{j}$, \cf Eq. \eqref{eq:outcomej}.

If the principle of superposition is true, how is it possible that the state $\ket{\tn{outcome }j}$ depends on the initial state $\ket{\psi}$?
In \cite{Sto17UniverseNoCollapse}, instead of using $\ket{\tn{outcome }j}$ to represent the measurement device as a low-level quantum system, it was proposed to use a macro state, which is the equivalence class of all micro states that are not distinguished by the observer. Then, to assume that the micro state, which cannot be perceived macroscopically, is the one that depends on the initial state $\ket{\psi}$.

If the missing difference from the conserved quantity is to be found in the environment, then there must have been transferred there through some interaction. But the interactions are included in the Hamiltonian, which is rendered irrelevant during the collapse. The reasonable way to explain how the interactions made sure that the state $\ket{\tn{outcome }j}$ depends on the initial state $\ket{\psi}$ is if this happened by unitary evolution alone, without collapse. Indeed, this is the path proposed in \cite{Sto17UniverseNoCollapse}, and which is used in this article. 

In addition to the above, in \cite{Sto17UniverseNoCollapse} it was shown that if we impose the spin conservation in the case of the spin measurements of spin $1/2$ particles, there should be no collapse of the spin degrees of freedom of the wavefunction.

This violation of conservation laws due to a discontinuous collapse is true for the measurement scheme in the standard interpretation of Quantum Mechanics \cite{Sto17UniverseNoCollapse}, but also for interpretations which take the wavefunction and its collapse as real, like the GRW interpretation \cite{GRW86,Ghirardi1990RelativisticDynamicalReductionModels}. But, recalling that even in the de Broglie-Bohm interpretation, we have to assign physical properties to the pilot wave rather than to the point-particle, they are affected too (unless we reject their reality \cite{DaumerDurrGoldsteinZanghi1996NaiveRealismAboutOperators,GoldsteinZanghi2013RealityAndRoleOfWavefunction}, which does not solve the problem of the conservation laws). It is sometimes claimed that in the Bohmian interpretation there is no collapse, the only thing that happens instead being that the measurement makes all of the branches of the wavefunction except one simply become ``empty''. Note that this emptiness is different from the update of information that the point-particle is in one of the branches and not the other, because if it was only about this information, then the empty branches would have been already empty before the measurement. But since before the measurement they have physical effects and interfere, and after the measurement they cease to have any effect, it is like they were there and then they disappeared.
Once the measurement is done, the empty branches are effectively collapsed by plugging into the guiding equations the resulting position of the observed particle. And once a branch of the wave is emptied, the mentioned physical properties are no longer attributed to that branch, but only to the one in which the Bohmian point-particle was detected. So Bohmian mechanics simply cannot avoid the wavefunction collapse and the problems resulting from this, including the violation of conservation laws \cite{Sto17UniverseNoCollapse}.
This also affects the Many Worlds Interpretation (MWI) \cite{Eve57,Eve73}. While MWI is unitary and without collapse as long as all branches or worlds are taken into account, collapse is still present within each of the branches. The fact that there is collapse within each branch implies that the result from \cite{Sto17UniverseNoCollapse} applies to each branch, hence conservation laws should be violated. Conservation laws are restored in MWI (and possibly in Bohmian mechanics) when all of the branches are taken into account, because unitary evolution remains valid.

It may be objected that such violations of the conservation laws are too small to be relevant, and that this should not count as a problem for the Copenhagen Interpretation or for other interpretations.
On the other hand, conservation laws are related to the symmetries, which play a central role in the foundations and the principles of the entire physics.
This argument may be seen as subjective, but there is an objective way to sort it out: let the experiment decide. Regardless of the personal taste, a prediction of the violations of the conservation laws have empirical consequences. An experiment is possible at least in principle, although it is hard to imagine how to make it, at least with the current state of technology.

It appears that the only way to satisfy Condition \ref{condition:conservation} is that the time evolution really is unitary, and the collapse is never real, but only apparent. It is indeed possible to have unitary evolution without collapse at the level of a single world, as explained in a series of papers by Stoica \cite{Sto08b,Sto12QMb,Sto12QMc,Sto13bSpringer,Sto16aWavefunctionCollapse,Sto17UniverseNoCollapse}. Note that previously Schulman made a proposal based on special states which, as initial and final conditions, are required to be separable, and evolve unitarily (see \cite{schulman1984definiteMeasurements,schulman1997timeArrowsAndQuantumMeasurement,schulman2012experimentalTestForSpecialState,schulman2016specialStatesDemandForceObserver} and references therein). This is not the position taken here and in Stoica's previous works, and the differences of motivation and implementation of the two approaches are discussed in \cite{Sto17UniverseNoCollapse}. In particular, the main difference consists of using global consistency between local solutions allowed by quantum measurements spread in spacetime, while Schulman's approach is based on the initial and final conditions of separability of states.

Another problem caused by the collapse of the wavefunction is that it leads, in some interpretations, to tensions with Condition \ref{condition:relativity_of_simultaneity}. The collapse theories \cite{GRW86,Ghirardi1990RelativisticDynamicalReductionModels} break this condition. The {\pwt} \cite{db1956double_solution,db1956double_solution_brief,Bohm52,DGZ1996BohmianMechanics}, even if they assume that the wavefunction never collapses, break it too. This seems to be a general problem when you try to have an ontological account of nonlocal effects. In {\pwt}, this may be related to the effective collapse of the conditional wavefunction, which corresponds to the non-empty branch. In Section \sref{s:no_collapse} we will see that the Many Worlds Interpretation breaks it too, although the way it breaks can be reduced significantly. There are only two known ways to avoid breaking it. One of them is to deny that anything happens between measurements is real, as in the Copenhagen Interpretation and its variants \cite{Wheeler1990InformationPhysicsQuantumLinks,Wheeler2000GeonsBHQuantumFoam,Rovelli1996RelationalQM}. Another way is to reject collapse and keep the {\schrod} equation valid only \cite{schulman1997timeArrowsAndQuantumMeasurement,tHooft2016CellularAutomatonInterpretationQM,Sto16aWavefunctionCollapse,Sto17UniverseNoCollapse}, which is the approach taken here.

Another problem with the wavefunction collapse appears when we take into account General Relativity. In this case, when the wavefunction collapses discontinuously, the stress-energy tensor operator $\hat T_{ab}$, and the expectation value $\langle \hat T_{ab} \rangle_{\ket{0}}$ of this operator collapses as well. This leads, in the semi-classical approximation, by using the following version of the Einstein equation
\begin{equation}
\label{eq:einstein_eq}
R_{ab} - \frac 1 2 R g_{ab} = 8\pi G \langle \hat{T}_{ab} \rangle_{\ket{0}},
\end{equation}
to a discontinuous change in the Ricci tensor $R_{ab}$, hence in the spacetime curvature. This means a discontinuity of the covariant derivative $\nabla$. But the covariant derivative is involved in the time evolution equations of all particles. This means that the wavefunction collapse of a single particle should affect the time evolution of all other particles whose wavefunctions propagate in the region where the collapse happened. If this effect would be testable, it could be used to send superluminal signals. So probably it cannot be tested, or no collapse actually happens.

In addition, in \cite{EppleyHannah1977NecessityQuantizeGravitationalField} it is shown that QFT on curved spacetime where equation \eqref{eq:einstein_eq} holds could be used to send signals faster than $c$ in a different way, and that it leads to violations of the uncertainty principle, assuming the wavefunction collapses discontinuously. 

In solving the problem introduced by the wavefunction collapse in the covariant derivative $\nabla$, an alternative to the solution offered by the elimination of discontinuous collapse is to quantize gravity. A quantum gravity theory could solve this problem for example by including a quantization $\hat{g}$ of the metric tensor $g$, along with a quantum version $\hat{\nabla}$ of $\nabla$, so that $\hat{\nabla}$ does not depend on the partial differentials of $\hat{g}$.
But it is perhaps too speculative to defer all of the problems to a future theory of quantum gravity.

All these arguments lead to the question whether it is really necessary that the wavefunction collapses in a discontinuous way, or whether is possible to avoid this and have unitary evolution even during quantum measurements.

%------------------------------------------------------------%
\section{What happens if the wavefunction collapse is unitary?}
\label{s:no_collapse}

The reasons mentioned so far justify us to at least consider seriously the possibility that the time evolution is always unitary. 

The unitary time evolution of quantum systems is not an additional assumption, it follows from the Schr\"odinger equation and its relativistic versions. What I do is not to add a new assumption, but to argue that the assumption that unitary evolution is suspended during measurements and replaced by a discontinuous collapse of the wavefunction is not actually proven by experiments, and it was accepted too quickly. The idea that the state vector projection needs to be discontinuous is in fact the new assumption, a radical one, never proved directly, and, I argue, unnecessary. If we can show that the discontinuous collapse is unnecessary, new possibilities open up.

But unitarity is a strong constraint, and its consequences have to be understood. Consider for example a measurement of the spin of a single particle. If the evolution is always unitary, and no discontinuous collapse happens, it means that the observed particle was already in a state which evolved in the observed eigenstate corresponding to the measured spin. We can account for the interaction between the observed particle and the measurement device, and such an interaction exists indeed, and changes the state of the observed particle. For example, the Stern-Gerlach device measures the spin of neutral atoms having a non-null magnetic moment by using the interaction between the atom and the magnetic field, so that the atom's spin changes. But the change is too small to bring the atom in an eigenstate of the spin along the chosen axis, if the previous spin of the atom was not already in an appropriate state which could evolve into the observed one. If the interaction taking place during the measurement is too large, then the Born rule is violated, in the sense that the spin can even be reversed from $\ket{\uparrow}$ to $\ket{\downarrow}$. Not any interaction counts as a quantum measurement.

The situation gets even trickier if we want to perform multiple non-commuting spin measurements on the same atom. The only way to do this without collapsing the spin in a discontinuous way is that the magnetic interactions with the two Stern-Gerlach devices are fine-tuned so that the atom leaves the first device with deviated spin, and is then deviated more by the interaction with the second device, so that the total deviation changes the spin from an eigenstate of the first spin operator to an eigenstate of the second one. This fine-tuning of course requires that the atoms in the two devices have fine-tuned states, as if they would conspire to give us the right outcomes.

If there is only unitary evolution, with no nonunitary collapse, then the total quantum state, containing the observed particle and the measurement device, should be in very special states, to allow for definite outcomes. This was shown in a theorem in \cite{Sto12QMb}. The theorem is so simple, that it may worth including  here a version of it. 

\begin{theorem}
\label{thm:unitary_ic}
Let $\hilbert_{\psi}$ and $\hilbert_{\mu}$ be the separable Hilbert spaces of the observed system, respectively of the apparatus (in which we consider the environment to be included).
Consider a Hermitian operator $\mc{O}$ on $\hilbert_{\psi}$, which is not a multiple of the identity (\ie has at least two distinct eigenvalues).
Fix a vector $\ket{\mu}\in\hilbert_{\mu}$.
Let $U:\hilbert_{\psi}\otimes\hilbert_{\mu}\to \hilbert_{\psi}\otimes\hilbert_{\mu}$ be a unitary operator (which represent the time evolution starting before and ending after the measurement). Suppose that for each eigenvalue $\lambda$ of $\mc O$ there is a vector $\ket{\psi}\in\hilbert_{\psi}$, an eigenvector $\ket{\psi'}$ of $\mc{O}$ corresponding to $\lambda$, so that
\begin{equation}
\label{eq:unitary_ic}
U\(\ket{\psi}\ket{\mu}\) = \ket{\psi'}\ket{\mu'}.
\end{equation}
Then, only a zero measure set of vectors $\ket{\psi}\in\hilbert_{\psi}$ satisfy \eqref{eq:unitary_ic} for some eigenvector $\ket{\psi'}$ of $\mc{O}$.
\end{theorem}
\begin{proof}
Let $\ket{\psi_1'}$ and  $\ket{\psi_2'}$ be two eigenvectors of $\mc{O}$ in $\hilbert_{\psi}$ corresponding to distinct eigenvalues, so that there are two vectors $\ket{\psi_1},\ket{\psi_2}\in\hilbert_{\psi}$ satisfying
\begin{equation}
U\(\ket{\psi_1}\ket{\mu}\) = \ket{\psi_1'}\ket{\mu_1'}
\end{equation}
and
\begin{equation}
U\(\ket{\psi_2}\ket{\mu}\) = \ket{\psi_2'}\ket{\mu_2'}.
\end{equation}

% Because $\ket{\psi_1'}$ and  $\ket{\psi_2'}$ are orthogonal, $\ket{\mu_1'}\ket{\psi_1'}$ and $\ket{\mu_2'}\ket{\psi_2'}$ are orthogonal too.
% Since $U$ is unitary, it follows that $U^{-1}\(\ket{\mu_1'}\ket{\psi_1'}\)$ and $U^{-1}\(\ket{\mu_2'}\ket{\psi_2'}\)$ are orthogonal too, and since they are just $\ket{\mu}\ket{\psi_1}$ and $\ket{\mu}\ket{\psi_2}$, it follows that $\ket{\psi_1}$ and $\ket{\psi_2}$ are orthogonal as well. 
Then,
\begin{equation}
U\(\(\alpha_1\ket{\psi_1}+\alpha_2\ket{\psi_2}\)\ket{\mu}\) = \alpha_1\ket{\psi_1'}\ket{\mu_1'} + \alpha_2\ket{\psi_2'}\ket{\mu_2'}
\end{equation}
for any two complex numbers $\alpha_1$ and  $\alpha_2$.
Suppose that
\begin{equation}
\alpha_1\ket{\psi_1'}\ket{\mu_1'} + \alpha_2\ket{\psi_2'}\ket{\mu_2'} = \ket{\psi''}\ket{\mu''}
\end{equation}
for some vector $\ket{\mu''}\in\hilbert_{\mu}$ and some eigenvector $\ket{\psi''}\in\hilbert_{\psi}$ of $\mc{O}$.
Since $\psi_1'$ and $\psi_2'$ are orthogonal, this is possible only if there is a number $\beta\in\C$ so that $\ket{\mu_2'}=\beta\ket{\mu_1'}$, so that
\begin{equation}
\(\alpha_1\ket{\psi_1'} + \alpha_2\beta\ket{\psi_2'}\)\ket{\mu_1'} = \ket{\psi''}\ket{\mu''}.
\end{equation}
Then, $\alpha_1\ket{\psi_1'} + \alpha_2\beta\ket{\psi_2'}$ is an eigenvector, for any $\alpha_1,\alpha_2\in\C$.
But this is only possible if $\ket{\psi_1'}$ and $\ket{\psi_2'}$ are eigenvectors of the same eigenvalue, leading to a contradiction.
Therefore, no linear combination $\alpha_1\ket{\psi_1}+\alpha_2\ket{\psi_2}$ can evolve into an eigenvector of $\mc{O}$, for all $\alpha_1\neq 0$ and $\alpha_2\neq 0$.
It follows that only a subset of measure zero from all possible initial states $\ket{\psi}$ can satisfy Eq. \eqref{eq:unitary_ic}.
\end{proof}

Theorem \ref{thm:unitary_ic} shows that the measure of the initial states which can result in definite outcomes of the measurements rather than Schr\"odinger's cat type of superpositions is zero compared to the measure of the entire Hilbert space. This result assumes sharp measurements, but if the measurements are not exactly sharp due to limitations coming from the \emph{Wigner-Araki-Yanase theorem} \cite{wigner1952MessungQMOperatoren,Wigner1952MessungQMOperatorenPBusch2010EnTranslation,ArakiYanase1960MeasurementofQMOperators}, still only a small subset of the Hilbert space can result in definite outcomes. In other words, the initial state of the observed particle has to be perfectly synchronized with that of the measurement device, even though they are separated initially by a spacelike interval. There is no way to escape the apparent conspiracy between the observed particle and the measurement device, even if they come from causally separated regions of spacetime.

Such fine-tuning of the initial state is usually interpreted in terms of retrocausality. Retrocausal models were already suggested by de Beauregard \cite{deBeauregard1977TimeSymmetryEinsteinParadox}, and after that by Rietdijk \cite{Rietdijk1978retroactiveInfluence}. Another approach, which provides a mechanism taking place in ``pseudotime'' is Cramer's \emph{Transactional Interpretation} \cite{cramer1986transactional,cramer1988overview,Kastner2012TransactionalBook}. Different models and approaches are proposed and discussed in \cite{SEP-2019qm-retrocausality,Wharton2007TimeSymmetricQM,price2008toyRetrocausality,sutherland2008causallySymmetricBohm,Argaman2008Retrocausality,price2015disentangling,sutherland2017RetrocausalityHelps,EmilyAdlam2018SpookyActionTemporalDistance,CohenCortesElitzurSmolin2019RealismAndCausalityI,WhartonArgaman2019BellThereomAndSpacetimeBasedQM}. An interesting and important proposal, based on evolution in both directions of time, is the \emph{two-state vector formalism} \cite{aharonov1964time,aharonov1991complete,cohen2016quantum2classical}.

A way to understand retrocausality is by appealing to Wheeler's delayed choice experiment \cite{Whe78}. In Wheeler's experiment, the setup is such that we can choose between making a which-path measurement or an interference measurement, after the moment when the observed photon either took both ways or randomly only one. So it looks like the photon's ``choice'' between the both-ways and the which-path possibilities is affected retrocausally by our choice of what experiment to perform. Of course, this experiment can be interpreted in terms of wavefunction collapse, but it is very eloquent in suggesting that a retrocausal effect takes place, which seems in this case a more natural explanation.

Later I will give an account of this apparent retrocausality, based on the block universe, which makes it less weird. But for the moment, let us face some more implications of this strange possibility.

An important feature of retrocausal approaches, including the purely unitary ones, is that they can recover both relativistic invariance and locality (or, we can interpret them as recovering locality in space, at the expense of ``locality in time''). They are consistent with relativistic locality, and what they violate is \emph{statistical independence}.
As we remember from Bell's theorem, there are two conditions leading to Bell's inequality \cite{Bell64BellTheorem,Bell2004SpeakableUnspeakable}. The first condition is that of locality, and the second one is that of statistical independence. Statistical independence means that the initial state of the observed pair of particles is independent from the experimental setup. From these conditions, Bell's inequality follows. But since the experimental evidence \cite{AGR82,Aspect99BellInequality} showed repeatedly that Bell's inequality is violated, it means that at least one of the conditions of the theorem is violated, but we cannot say for sure which one.
Retrocausal approaches violate statistical independence, and save locality.
There is another implicit assumption in Bell's theorem, that the outcomes are determinate. If we assume that all outcomes happen, we obtain the Many Worlds Interpretation. I will come back to this later, because it provides a nice way out without having to choose between locality and statistical independence, and it also brings us a step closer to the unitary approach presented in this article. For the moment, we need to say more about retrocausality and locality.

Most physicists and philosophers of physics find it more acceptable to give up locality, and to maintain statistical independence. The reasons are obvious, violations of statistical independence seem to violate causality.

But a closer look considering relativistic causality in Minkowski spacetime shows that both options have similar problems. Nonlocality, coined by Einstein as ``spooky action at a distance'', is at odds with relativistic causality. But this nonlocality does not allow us to send signals or energy outside the light cone by quantum measurements. This is considered to be consistent with relativity, of course, at an operational or epistemic level, not at the level of ontology. At the epistemic or operational level there is a ``peaceful coexistence'' between QM and SR, but if we want to avoid being instrumentalists about QM, we should not do so by becoming  instead instrumentalists about SR, so I would argue that no-signaling is not enough, locality should be obtained at the ontological level.

But the same can be said about retrocausality. It does not violate causality, because it cannot be used to send information or energy back in time. It cannot be used to change the outcomes of already performed measurements. The experiments of quantum time travel \cite{SLloyd2011QuantumTimeTravel} can be done, but there is no such observable violation of causality, just like in the quantum teleportation experiment no such observed violation of Einstein causality occurs \cite{bennett1993teleporting,vaidman1994teleportation,bouwmeester1997ExperimentalQuantumTeleportation}. 
But if we take seriously Einstein's causality on Minkowski spacetime, and especially on curved spacetime, the option of rejecting locality is arguably more problematic than rejecting statistical independence. Both Special and General Relativity have well-defined ontologies, which are local. On a curved background, the fields whose stress-energy tensor corresponds to a spacetime curvature via Einstein's equation should have well-defined local ontologies too. So it is irrelevant from this point of view that we can claim that violations of causality happen at the ontological level, but they are not observable. If we seriously adhere to the goal of providing an ontology for Quantum Mechanics, we have to take the ontology seriously all the time, and not conveniently ignore this ontology on the grounds that no quantum measurement can show that it violates Einstein's causality. But the alternative way, of an ontology based on violations of statistical independence which does not violate locality, is not at all inconsistent at the ontological level with Einstein's causality, especially in the block universe view.
Retrocausal models have an advantage -- ontologically, they are perfectly consistent with SR and GR, there is no need to invoke instrumentalist arguments to recover this consistency. The only reason to invoke no-signaling backwards in time is to recover causality, but this kind of causality is not necessary for SR or GR at a fundamental level, it is only needed at the coarse grained level.
Indeed, both SR and GR are perfectly time symmetric theories, and they cannot distinguish past from future lightcones, so the covariant equations can be solved in both directions of time, and they allow influences in both directions of time. Not only are they consistent with bicausal interpretations, but they even endorse them. A direction of time is preferred because of the Second Law of Thermodynamics, but this preference is only statistical, and for the other laws retrocausal influences can always be reinterpreted as causal influences with special initial conditions.

Therefore, retrocausal interpretations provide a way to make QM consistent with Conditions \ref{condition:relativity_of_simultaneity}, \ref{condition:weak_causality}, and \ref{condition:free_will}, but they appear to be in conflict with Condition \ref{condition:strong_causality}.
The retrocausal models consistent with Condition \ref{condition:free_will} are those in which the evolution is not deterministic, \eg the Transactional Interpretation. For example, the EPR-B experiment is understood in the following way. The past is determined up to the moment when the singlet state decays. Then, the way it decays is determined completely only when Alice and Bob finish their measurements. 
So, in this case, retrocausality goes back only to the moment of decay.
Then, it is possible to explain the results as the particles not being actually in entanglement, but rather as if the singlet state decayed into two separate states, as if it already collapsed before the two particles went in separate places \cite{deBeauregard1977TimeSymmetryEinsteinParadox,Rietdijk1978retroactiveInfluence,price2015disentangling}. This interpretation of the EPR-B experiment seems to be confirmed through weak measurements \cite{aharonov2012future-past} in the exact way that Bohmian trajectories are considered to be confirmed \cite{Kocsis2011ObservingTrajectories}. Hence, we can consider the processes taking place during the EPR-B experiment as being local in the sense that the particles are described by local solutions of the Schr\"odinger equation. This kind of spacetime locality is not what we usually expect when we speak about locality, because it depends on the final conditions imposed by the experimental setup. The solutions are local in the sense that they obey partial differential equations on spacetime, but they are also subject to boundary conditions which are global and impose the apparent (spatial) nonlocality like that from Bell's theorem.

A retrocausal interpretation may also be deterministic, but sometimes this case is called superdeterminism (\eg \cite{tHooft2016CellularAutomatonInterpretationQM}). There are two ways to interpret superdeterminism. One of them makes it look worse than it is, since it assumes that the past state is completely determined, and the future experimental setups that would contradict the past are not allowed, violating thus Condition \ref{condition:free_will}.

The present proposal, being unitary, is deterministic, and because of Theorem \ref{thm:unitary_ic}, it must be of the same kind as superdeterministic models.
But we will see in Section \sref{s:global_consistency} that this is not as bad as expected, since the block universe allows for another way to interpret the model, which is consistent with both Conditions \ref{condition:strong_causality} and \ref{condition:free_will}.

Moreover, unitarity indeed avoids the problems of nonlocality, but we have first to make sure that the unitary evolution ontology is local. As I explained, the wavefunction is defined on the configuration space, rather than on the physical space. This means that the wavefunction is \emph{holistic}. But even so, it admits a representation in terms of fields on spacetime, which propagate and interact locally, as long as the evolution is unitary \cite{Sto2019RepresentationOfWavefunctionOn3D}. Being holistic means that it allows entanglement, but since there is no wavefunction collapse, this entanglement is never projected to separable states. If this would have happened, then of course violations of locality would occur, because such a projection would have to be nonlocal.

To understand how unitary evolution is local, let us first consider a wavefunction which is a separable state in a basis of eigenstates of an operator which commutes with the Hamiltonian. Since the relativistic evolution equation of such a state does not contain nonlocal interactions, this means that the state will evolve in a local manner, even though it is a multiparticle state. A general state is a superposition of such states, and each term of the superposition evolves locally. And since we cannot project any of them out, because we assume that all evolution is unitary and no discontinuous collapse occurs, this means that locality is ensured. Regardless of what other meanings we assign to locality, this type of locality is consistent with Einstein's locality and causality. More about how the wavefunction can be represented in a locally separable way on the three-dimensional space, and how its unitary evolution is local, can be found in full mathematical detail in \cite{Sto2019RepresentationOfWavefunctionOn3D}.

Since I mentioned Einstein's locality and causality, it would be interesting which of the conditions from the 1935 EPR paper \cite{EPR35} is violated by this proposal. It is not difficult to see that Einstein's criterion of reality is violated, but not in a ``lethal'' way. Indeed, our ontology does not satisfy Einstein's realism, because the state before the measurement cannot be just any state, it depends on the experimental setup. This is another way to say that it violates statistical independence. Simply put, you cannot have any initial state for the observed particle and at the same time any initial quantum (microscopic) state of the apparatus. But this is not ``lethal'', since it is perfectly consistent with relativity and Einstein's causality and locality.
In addition, as argued in \cite{Sto17UniverseNoCollapse}, the macroscopic state of the apparatus and the quantum state of the observed system can be statistically independent.

Moreover, the avoidance of a discontinuous collapse also avoids the problems identified in Section \sref{s:collapse_problems}, while a nonlocal ontology, at least in the currently known forms, cannot do this, because there is no way to take the wavefunction as becoming empty of physical properties, and at the same time to say that this is not the same as collapse, as explained in Section \sref{s:reality_beyond_wavefunction}.

But, even if we still find retrocausality preposterous, as I already alluded, there is a way to tame it, which is based on the fact that there is another hidden assumption in Bell's theorem, which requires that the outcomes are determinate. The violation of this assumption makes possible the Many Worlds Interpretation, which allows all outcomes to be obtained. Alternatively, one can also interpret MWI at the branch level, where the outcomes are determinate, as violating one of the two main assumptions of Bell's theorem. At first sight, it may seem that branching happens through some nonlocal collapse accompanying the measurement, and then in each branch locality is violated. However, the other option makes more sense. Consider the EPR-B experiment, and start from the outcomes obtained by Alice and Bob. When the two particles are detected, they no longer form an entangled state, and they no longer interact. Therefore, if we evolve the solution unitarily backwards in time, the two particles turn out to be separable and to not interact for the entire time interval between the decay that produced them and the measurement, so no collapse is required to occur except when the spin $0$ particle decayed into two spin $1/2$ particles. So it makes sense to consider that the collapse accompanying the branching took place, in a very localized manner, when the decay occurred, and the particles evolved unitarily until their detection. This makes some proponents say that MWI is local \cite{Bacciagaluppi2002LocalityInMWI,Deutsch2011VindicationOfQuantumLocality,SEP-Vaidman2002MWI,Vaidman2016AllIsPsi}. A salient feature of this local interpretation of MWI is that, although in each branch retrocausality seems to happen, this does not really happen in the universal wavefunction, so overall, there is no violation of either statistical independence or locality, because all possible outcomes occur.

We have seen that it makes more sense that, in MWI, the split happens during decays and when systems are detected. While each branch taken independently manifests, by this, retrocausal features, on the overall there is no such thing. In the example with the EPR-B experiment, the branching seems to happen only during the decay into two spin $1/2$ particles, and this can be understood by pushing unitary evolution back in time as much as it holds without contradicting the observations. But is it necessary to stop this at the decay? As it is understood from the Weisskopf-Wigner model of decay \cite{WeisskopfWigner1930}, while an isolated atom or composite particle should be stable, even if it is excited, vacuum fluctuations make it unstable, and this can be described as a superposition between the excited state and the decayed state (including the products of the decay). Then, the decay happens as a result of the amplitude damping of the excited state, so eventually the system decays. Observations of the system may introduce a collapse, which either makes the system decay, or makes it continue to remain in the excited state. But in the case of the singlet state used in the EPR-B experiment, there are different ways to decay, and the vacuum also is in an undetermined state. So it is possible that, when evolving backwards in time the two particles, one obtains a history in which the vacuum was in the right state required to provoke the decay of the singlet state exactly as needed for the subsequent measurements made by Alice and Bob to find the two particles in the right states. It is, at least in principle, possible that there is no actual branching at the fundamental level, but at the coarse-grained macroscopic level it would appear as if there is a collapse and that there is randomness. In this case, the MWI would be modified into a theory in which the worlds, rather than splitting, are really parallel and evolve unitarily, but they are not always distinguishable at the coarse grained level. In addition, it is known that MWI has a problem with probabilities. While even since Everett's seminal work \cite{Eve57,Eve73} it was proposed to attribute probabilistic meaning to the amplitudes, this does not seem very appropriate, because the wavefunction is ontic. It would be perhaps more convincing to have a derivation of probabilities in terms of ensembles of microstates. But in a theory of truly parallel worlds, which only seem identical up to some point when measurements happen, and then they become distinguishable, probabilities would occur naturally in terms of ensembles of microstates. So, if we take the idea of parallel worlds which always evolve unitarily, we recover the conservation laws for each of the worlds, and also make room for probabilities. From all these arguments, it seems that if we push MWI to its consequences, each branch should in fact be independent and evolve unitarily for its entire existence. It remains an open problem how to derive the correct probabilities according to the Born rule in such a theory.

%------------------------------------------------------------%
\section{Global consistency condition}
\label{s:global_consistency}

We have seen that, if the unitary evolution is not broken even during the measurements, then the initial states including both the observed system and the measurement apparatus (and the environment), have to be severely constrained, otherwise the outcome of the measurement will be undefined \cite{Sto12QMb}. This looks like fine-tuning, retrocausality, or superdeterminism. We have seen in section \sref{s:no_collapse} that MWI provides a nice way to tame retrocausality. But here I will discuss another possible explanation, which works for a single world, and makes sense in the block universe picture, and only looks like fine-tuning or retrocausality for the observers experiencing the flow of time.

For this, we need to take into account two different but related dynamical systems
\begin{enumerate}
	\item 
	The micro quantum dynamical system, consisting of the unitary evolution in a Hilbert space, governed by the {\schrod} equation.
	\item 
	The macro-quasiclassical dynamical system, which is the coarse graining of the low-level quantum dynamical system.
\end{enumerate}

Here is a definition of a dynamical system, which works for both deterministic and nondeterministic systems.
\begin{definition}
\label{def:dynsys}
A \emph{dynamical system} $\mc{A}=(S,T,\tau)$ consists of 
\begin{enumerate}
	\item a \emph{state space} $S$,
	\item a one-dimensional group or monoid $T$, which can be the real numbers $\R$ or the integers $\Z$, or the non-negative reals $\R^+$ or integers $\N$, and represents \emph{time},
	\item an \emph{evolution law} $\tau:T\times S\to\ms{P}(S)$, where $\ms{P}(S)$ is the set of all subsets of $S$, and $\tau$ satisfies the conditions $\tau(t_2,\tau(t_1,s))=\tau(t_1+t_2,s)$, $\tau(0,s)=\{s\}$, for all $t_1,t_2\in T$ and $s\in S$.
\end{enumerate}
\end{definition}

For a deterministic dynamical system, the evolution law maps states into states, but in Def. \ref{def:dynsys} it maps states into sets of states, in order to be able to describe nondeterministic systems as well. In the deterministic case, the mapping is always to sets consisting of only one state.

An example of dynamical system of $\n$ particles is any classical mechanical system
\begin{equation}
\label{eq:dynsys_classical}
\mc{C}=(S_\mc{C}=\R^{6\n},T_\mc{C}=\R,\tau_\mc{C}),
\end{equation}
where the state space is the phase space $S_\mc{C}=\R^{6\n}$ consisting of the positions and the momenta of the $\n$ particles, the time is $T_\mc{C}=\R$, and the evolution law $\tau_\mc{C}$ is given by Hamilton's equations. In statistical mechanics, one uses a \emph{coarse graining} of the classical phase space, which is a partition of the classical phase space into disjoint subsets. It is assumed that there is an equivalence relation between the points in the phase space representing classical states, defined for any $s_1,s_2\in S$ by $s_1\cong s_2$ iff $s_1$ and $s_2$ cannot be distinguished at the macro level. This relation is in fact not precisely defined, but it works very well to reduce thermodynamics to classical statistical mechanics. We will denote the macro classical dynamical system by
\begin{equation}
\label{eq:dynsys_classical_macro}
\mc{M}=(S_\mc{M},T_\mc{M}=\R,\tau_\mc{M}).
\end{equation}

For a quantum dynamical system, we will use here the nonrelativistic quantization, which consists of starting with a classical dynamical system $\mc{C}=(S_\mc{C},T_\mc{C},\tau_\mc{C})$ of $\n$ particles, and replacing the states with square integrable complex functions on the configuration space $\R^{3\n}$. The {\schrod} equation is obtained by quantization from the classical evolution law.
The resulting state space is a Hilbert space $\hilbert$, the time is $T=\R$, and the evolution law is the unitary evolution defined by the obtained {\schrod} equation as in Eq. \eqref{eq:unitary_evolution},
\begin{equation}
\label{eq:dynsys_quantum}
\mc{Q}=(\hilbert,T_{\mc{Q}}=\R,\tau_\mc{Q}=\hat{U}).
\end{equation}

We will also denote the macro-quasiclassical dynamical system by $\mc{M}=(S_\mc{M},T_\mc{M},\tau_\mc{M})$.

To better see the relation between the classical dynamical system of $\n$ particles $\mc{C}$ from Eq. \eqref{eq:dynsys_classical} and the quantum dynamical system $\mc{Q}$ from Eq. \eqref{eq:dynsys_quantum}, it is useful to remember the \emph{phase space formulation of QM}, introduced in \cite{Weyl1927QuantenmechanikUndGruppentheorie,Wigner1932PhaseSpaceQM} and developed in \cite{Groenewold1946PhaseSpaceQM,Moyal1949PhaseSpaceQM}.

The phase space formulation is equivalent to the standard formulation, because the Wigner-Moyal transform has the same properties as the tensor product. A key difference in connecting the macro level with the quantum level is that, in the phase space formulation, the relation between with classical coarse grained system $\mc{M}$ is more visible.
As such, all observations consist of observations of the macro states. Quantum measurements are arrangements of the macro states that allow us to infer information about the quantum state from observations of the macro states. What we learn is in what cell of the coarse graining is the phase-space function localized. In the standard formulation based on configuration space, what we learn is in what subspace of the Hilbert space is the state vector contained.

Consider now a collection of all observations made so far, whether they consist of free observations of the world, of classical, or of quantum measurements. Recall that what is observed are just the macro states. In fact, in practice we observe subsystems of the dynamical system $\mc{M}$, which allow us to restrain the set of all possible macro states from $S_\mc{M}$, but not to determine them completely, in the absence of information about the entire system. But let us assume for simplicity that we can observe the macro states. Our observations are a collection of macro states realized at different times $\ldots, t_{-1},t_0,t_1,\ldots$. These observations allow us to infer more and more about the micro, quantum state. In other words, we start with a set of possible quantum states, and with each observation at a given time $t_j$, we reduce the set of possible solutions more and more, as we refine our observations and we learn more about the state of the universe. In terms of phase space, at each time $t_j$ we observed a coarse graining cell in the phase space. In terms of the Hilbert space, at each time $t_j$ we observed a subspace of the Hilbert space. Then, finding out the quantum state based on these observations, and by taking into account its dynamics, consist of something similar to fitting a curve. What we do is to fit the trajectory of the quantum system in the Hilbert space, given some constraints at various times $\ldots, t_{-1},t_0,t_1,\ldots$.

We have the following question. Is it always possible to fit the constraints without breaking the unitary evolution? If we focus on the observed system only, it seems that we need to introduce the collapse of the wavefunction, at least for the cases when it is measured using non-commuting operators. But given the discussion we had so far, it is possible to have solutions in which the observed system, but also the micro level of the apparatus, are in such a state that the interaction at the micro level takes place in the right way to account for our macro level observations. In \cite{Sto17UniverseNoCollapse}, a thought experiment was proposed, to conclude that there are always such solutions that fit all of the macro constraints by unitary evolution alone. Also see the discussion about the abundance of the ``special states'' in \cite{schulman1997timeArrowsAndQuantumMeasurement}, where it was argued that there are always such solutions, in fact even more restricted, special quantum states.

An omniscient mind that would know in complete detail the micro quantum level $\mc{Q}$ would be able to predict all of the measurements and their outcomes, by applying the unitary evolution. For such a mind, the world would appear superdeterministic, and it would satisfy Condition \ref{condition:strong_causality}. On the other hand, for being like us, even as embedded in the very system we observe, everything happens at the macro, coarse grained level $\mc{M}$, and it satisfies Conditions \ref{condition:weak_causality} and \ref{condition:free_will}. So there will be no violation of causality or of free will if we accept a top-down retro-causation from the macro level to the micro level. But there is another perspective, the timeless perspective of the block universe. This has to be post-determined, in the sense of fitting the unitary evolution curve after all of the macro constraints are collected.

To arrive to that perspective, we need to acknowledge that, while the dynamical system description is complete, it hides something, or rather it makes it implicit. What is not obvious in the dynamical system picture is the spread of constraints in spacetime. Maybe another formalism can complement the dynamical system description, by taking into account the fact that various observations happen in different places in space, in a local manner.

Such a description is provided by \emph{sheaf theory} \cite{MacLaneMoerdijk92,Bredon1997SheafTheory}. For a sheaf theoretical framework for physics and foundational problems also see \cite{Sto08WorldTheory}.
Sheaf theory has many applications, but the one in which we are interested is how local solutions of partial differential equations (PDE) extend globally. When a PDE has well-defined initial conditions which ensure local solutions, this is not necessarily enough to ensure global solutions.

A \emph{sheaf} is a collection of objects defined locally on open subsets of the base space, which satisfy certain conditions that allow them to be ``glued'' together into solutions defined on larger sets in a nice way. For example if $U$ and $V$ are two open sets so that $U\cap V\neq\emptyset$, and if there are two fields $u$ and $v$ defined locally on $U$, respectively $V$, so that the restrictions $u|_{U\cap V}=v|_{U\cap V}$, then the fields $u$ and $v$ can be glued to a field $w$ defined on the union $U\cup V$. In addition, if $u$ is a local field on $U$, and $U'\subset U$, then the local solution $u_{U}'$ also belongs to the sheaf. There are additional conditions to be satisfied, but they are automatically satisfied in the cases of interest to us, which come from PDE.

In the particular case of local solutions to PDE, which is of interest for us, the conditions defining a sheaf are automatically satisfied. When two local solutions have disjoint domains, or are equal on the common domain on which they are defined, they can be ``glued'' and extended to the union of their individual domains. 

But this extension of local solutions does not always work, because global extensions are not always guaranteed to exist. Even if the domains of the local solutions do not overlap, it is not always possible to have a global solution extending them. By this, sheaf theory shows us that, when the information we have about the wavefunction is spread in multiple locations, not all local conditions can be mutually consistent, so there are certain constraints, and correlations are enforced between local solutions. This is relevant for our discussion, because it suggests a way to limit the total Hilbert space in Theorem \ref{thm:unitary_ic} to a subset of solutions which lead to definite outcomes by unitary evolution without collapse.

Recall from Section \sref{s:wavefunction_on_spacetime} that the wavefunction can be represented as a vector field on the $3$D-space, according to \cite{Sto2019RepresentationOfWavefunctionOn3D}. Moreover, the representation in \cite{Sto2019RepresentationOfWavefunctionOn3D} satisfies local separability, which means that the vector fields representing the wavefunction can be defined locally, and they indeed form a sheaf.
Two things are important here: (1) that the local solutions form a sheaf, and (2) that the base space is the physical $3$D-space, which allows the connection with measurement devices, which are macroscopic systems, so they are objects on the physical $3$D-space.

For example, if a particle is found at a certain location, its amplitude at other locations should vanish. Similarly, in the EPR-B experiment, the local solutions found by Alice and Bob have to be mutually consistent, even if they are imposed at different locations. This is where sheaf theory becomes useful.
We know that not all combinations of outcomes that Alice and Bob can obtain are consistent, and those that are consistent have different probabilities to occur. Global extensions of sheaves have similar properties.

\emph{Sheaf cohomology} studies the obstructions which prevent the extension of local solutions to global ones. These obstructions are usually of topological nature. The space of initial or boundary conditions which admit global extensions is reduced, in the presence of such obstructions, to lower dimensional spaces. 
A simple example is that of \emph{complex holomorphic functions}. They are complex functions satisfying the \emph{Cauchy-Riemann} condition, which is equivalent to saying that such functions depend on $z\in\C$ but not on its conjugate. On the complex plane $\C$ they are spanned by the powers $n\geq 0$ of $z\in\C$, but the solutions do not always extend globally through analytic extension, they can run into singularities. Those with global extensions on $\C$ without singularities are called holomorphic, and their space is smaller than the space of functions spanned by the powers of $z$, yet still infinite-dimensional (the polynomials in $z$ form an infinite-dimensional subspace of the space of holomorphic functions on $\C$). But on the Riemann sphere $\C\cup{\infty}$, the only holomorphic functions are the constant ones, so their space is one-dimensional. Changing the topology by adding a single point introduces incredible constraints.
Such examples suggest an interesting possibility: what if there are constrains, perhaps topological in nature, at the level of the spin and gauge bundles and especially interaction bundles used to describe particles and forces, which ensure that only certain solutions are possible? Could the global solutions correspond to the definite outcomes of measurements, and explain the correlations between outcomes obtained when measuring entangled systems? From the mathematical examples we have in sheaf theory, this appears to be a plausible possibility, and it makes sense to explore it as a candidate solution of the measurement problem and the problem of the apparent classicality at the macroscopic level. But it is difficult to know for sure, until we will have a better understanding of the fiber bundle structures, and of what kind of constraints they impose to the global existence of solutions.

One may think that the example of complex holomorphic functions is irrelevant to Quantum Mechanics for two reasons. First, what is the relevance of the Cauchy-Riemann condition? Well, the Cauchy-Riemann operator is for two real dimensions what the Dirac operator is for Minkowski spacetime, as it is known from the theory of Clifford algebras \cite{chevalley1997algebraicspinors,crumeyrolle1990clifford}. Moreover, the same operator can be used to express the Maxwell equations in a more compact way as a single equation \cite{hes:1966}, and this also works for the Yang-Mills equation.

A second objection could be that Minkowski spacetime has a trivial topology, and even the curved spacetime probably has, if not a trivial topology, maybe a simple one. And this is true, but on the other hand in reality both the existence of spin $1/2$ particles, and of the internal gauge degrees of freedom, require the existence of fiber bundles, and they have a complicated topology.
Unfortunately, it is not currently understood exactly what happens from a topological point of view, and consequently we do not know the conditions to be satisfied for a solution to be global. But the point is that the Hilbert space is severely reduced, and we need to see exactly how. Future understanding of the geometric and topological properties of particles and their interactions, especially in a quantized theory, may shed light on this issue and explain both why the quantum states appear classical at a macroscopic level, and how quantum measurements yield definite outcomes without invoking the wavefunction collapse.

It is interesting that Schr\"odinger used a consistency condition on the space to obtain the energy eigenstates of the Hydrogen atom, in the form of a boundary condition of the wavefunction of the electron at infinity \cite{schrod:1926}. The usual view on this is in terms of eigenstates of the Hamiltonian, but in terms of a wavefunction on space, this turned out to come from boundary conditions at infinity.
The requirement of global consistency is similar, but on spacetime rather than space, in fact, on the spinor and gauge bundles.

Unfortunately, at this stage such a solution is still speculative, in the absence of a better understanding of the topological constraints of the fields due to the topology of the fiber bundles or other reasons.

However, we can postulate such constraints and obtain an effective description.
\begin{principle}[of global consistency]
\label{principle:constraints}
The macro laws of the dynamical system $\mc{M}$ constrain the allowed solutions of the micro quantum dynamical system $\mc{Q}$.
Only those solutions of the {\schrod} equation that are consistent with the observations at the macro level $\mc{M}$ are accepted.
\end{principle}
The laws of $\mc{M}$ are effective laws, most of them obtained from the classical system $\mc{C}$, as well as empirical. For example since we do not observe macro {\schrod} cats or superpositions of states of the measurement device indicating different outcomes, we may assume that there is an effective law constraining the solutions of the {\schrod} equation to prevent such ``grotesque states''.

Principle \ref{principle:constraints} is consistent with Condition \ref{condition:universal_law} that the dynamics is given by a universal law, which is the {\schrod} equation. Unlike collapse, it does not break the evolution, but it restricts the admissible solutions of the {\schrod} equation.

But what is relevant to our discussion is that, if this is the case, then we have an explanation for the restrictions on the initial conditions: only those initial conditions leading to globally consistent solutions on the entire spacetime are admitted. And while such severe restrictions appear as a conspiracy to an observer experiencing the flow of time, a bird's eye view of the block universe would see everything just as the natural condition that the solutions are global.

Hence, despite the fact that the block universe is sometimes seen as being at odds with Quantum Mechanics, it complements it and offers a possible solution to its problems, by allowing us to think at the observations as constraints spread in spacetime, which have to be satisfied by the low-level solution of the {\schrod} equation.

The solutions of Schr\"odinger's equation are unitary, but when we think about ``wavefunction'', we think at two different things. 
On the one hand, as long as no measurement is made on a quantum system, we can regard the wavefunctions as a function on the configuration space, but it also admits a representation as fields on spacetime \cite{Sto2019RepresentationOfWavefunctionOn3D}. On the other hand, no measurement can completely determine the quantum state of the entire system made of the observed system and the measurement apparatus (and other relevant parts of the environment). This means that our measurements cannot be used to determine the future outcomes of the measurements of the observed system, but only probabilities. I do not yet have a proof whether these probabilities are exactly those given by the Born rule or not, but the possibility to recover the Born rule exists. The real, ``ontic'' wavefunction will be never completely determined, but what we can know is an ``epistemic'' wavefunction. The notions ``ontic'' and ``epistemic'' may be used differently by different authors, but I will stick with the definition that there is a real, ``ontic'', physical wavefunction, and ``epistemic'' is the partial knowledge of the ontic wavefunction of the universe, which translates as probabilities when it comes to learn more about it through measurements.

Some usual terms associated to the wavefunction have a statistical connotation: expectation value, uncertainty, \etc. These terms retain their statistical meaning when we are talking about the epistemic wavefunction, which is probabilistic. But when we are talking about the ontic wavefunction, the ``expectation value'' of an operator $\hat{\mathcal O}$ simply the mean value field $\bra{\psi}\hat{\mathcal O}\ket{\psi}$, and similarly for uncertainty and other terms.

In particular, 
\begin{equation}
\label{eq:stress_energy_expectation_value}
\langle \hat{T}_{ab} \rangle_{\ket{0}} = \bra{0}\hat{T}_{ab}\ket{0}
\end{equation}
is not a true expectation value, but, after suitable regularization, a physical field on spacetime, which consists of mean values. This will be relevant in the next section.

%------------------------------------------------------------%
\section{Quantum Theory and General Relativity}
\label{s:quantum_and_relativity}

Since the intention of this paper is to explain how the block universe is consistent with QM and can even help to solve some of its problems, it is necessary to analyze the tensions between QM and GR. In particular, we need to discuss the necessity of a theory of quantum gravity, how this may work, and what will survive from QM and GR. In particular, it would be interesting to see how Conditions \ref{condition:qeinstein}, \ref{condition:qg}, and \ref{condition:bh-info} can be satisfied.

Despite the measurement problem and the problem of the emergence of the classical world out of Quantum Mechanics, Quantum Field Theory is incredibly successful in describing the microscopic scale. At large scales, General Relativity does the job with equal success. Both theories were extensively tested, with great precision and in diverse situations, and their predictions turned out to be accurate every time. Not everything is understood, for example in cosmology there are the problems of inflation, dark energy, and dark matter. We do not now yet if they require changing the Standard Model, Quantum Theory, or General Relativity.

However, since the essential parts of both theories have to be true, we need to understand how they work together. When we try to combine them, they seem to be in a conflict. The general trend is to consider that one of the two theories will have to be radically changed, or even replaced, and that this one is GR. The main invoked reasons are the successes of Quantum Theory, the prediction of singularities in GR, the black hole entropy, and the information loss paradox. But mainly the resistance of gravity to be quantized in the same way as the other forces.

The predictions of QFT are confirmed to great accuracy, in particular in the case of the anomalous magnetic dipole moment. On the other hand, GR's predictions are confirmed with at least the same accuracy, and even better in the case of the Hulse-Taylor binary pulsar PSR 1913+16 \cite{HP96}. I think both theories were confirmed incredibly well in all predictions that could be tested, so it would not be fair to hold the success of QFT against GR while ignoring its own success.

It is true that GR predicts the occurrence of singularities \cite{Pen65,Haw66i,Haw66ii,Haw67iii,HP70}. But QFT is plagued with infinities too, in both the UV and IR regimes. It is true that renormalization worked particularly well in the predictions of great precision, and the \emph{renormalization group} (actually it is a semigroup) idea provides a deeper understanding, but it is not as if the infinities are completely cured. In fact, they are rather an artifact of the perturbative expansion. But the singularities in GR are not worse at all. Indeed, it is more difficult to see this, but it turns out that differential geometry \cite{Sto11a}, and in particular GR \cite{Sto12b}, can be formulated in a completely invariant way which is equivalent to the standard formulation outside singularities, but which allows us to write field equations, including an equation equivalent to Einstein's, even at singularities (see \cite{Sto13a} and references therein).
Moreover, the same solution of the problem of singularities was used as an ingredient in an explanation of the observed values of the expansion rate of the universe \cite{unruh2017dark_energy,unruh2019dark_energy}.

As for the problem of the quantization of gravity, do we really need to do in the same way as other forces are quantized? They are of apparently different nature, gravity is due to spacetime curvature, and the other forces are gauge fields. There seems to be no \emph{apriori} reason to do the same for gravity, which is just inertia on curved spacetime.
Could the theory of quantum fields on curved background, where the expectation value of the stress-energy operator is introduced in the Einstein equation \eqref{eq:einstein_eq}, be enough? This is usually called \emph{semi-classical gravity} \cite{LichnerowiczTonnelat1959theoriesRelativistesGravitation,rosenfeld1963quantizationOfFields}, and is considered to be an approximation of the true quantum gravity which remains unknown. But could it be enough?

In the case of a single particle which never collapses, it seems at first sight that there is no problem with semi-classical gravity: even if the particle is set in various superpositions, what matters for GR is its stress-energy tensor. If there are more particles, or even an undefined number as it is often the case in QFT, particularly on curved spacetime, the expectation value of the stress-energy operator does the same job. And since there is never a wavefunction collapse even though the macroscopic world seems classical and the measurements have definite outcomes, the stress-energy tensor is conserved and behaves as any classical source for gravity and spacetime curvature. This is made even more plausible by the result that the wavefunction and the QFT states are equivalent to classical fields, discussed in Section \sref{s:wavefunction_on_spacetime}.

However, several arguments seem to suggest that semi-classical gravity is not enough.

A known argument was given by Eppley and Hannah in \cite{EppleyHannah1977NecessityQuantizeGravitationalField}, where it is argued that if gravity is classical, one could use it to measure a quantum state  in a way which violates the uncertainty principle. Even though their argument was refuted \cite{Mattingly2006EppleyHannahFail,HuggettCallender2001WhyQuantizeGravity,Albers2008MeasurementAnalysisAndQuantumGravity,Kent2018SimpleRefutationEppleyHannah}, it worth discussing it here, especially because it was not discussed in conjunction with the unitary interpretation presented in this article, which introduces new subtleties.

Eppley and Hannah  identify two cases:
\begin{enumerate}
	\item 
If the wavefunction of the observed system is collapsed during the classical measurement with a gravitational wave, such a measurement could lead to violations of the momentum conservation.
	\item 
If the classical measurement does not collapse the wavefunction of the observed system, then we could use it to send superluminal signals.
\end{enumerate}
To prove this argument, one considers a particle in two boxes, one given to Alice, and another one to Bob. Then, Bob can ``look'' inside his box using arbitrarily non-disturbing classical measurements, and determine the state of the wavefunction in his box without collapsing it. By this, he would be able to know if Alice looked in her box. So, it is concluded, semi-classical gravity may lead to the transmission of signals faster than $c$.

But if there is no wavefunction collapse, option (1) is easily avoided without quantizing gravity.
At first sight it may seem that option (2) is avoided as well, because even if Bob does not collapse the wavefunction, Alice does it.

However, in our approach, measurements can introduce challenges. We have seen that if quantum measurements take place by unitary evolution, the present state of the particle should depend on the future measurement. So, if we can use such weak classical gravitational fields to determine the state of a particle without disturbing it too much, we can extract information about the choice of the future quantum measurement and its outcome. Since this choice is supposed to be free, it can be used to send information back in time.
It can be argued that this situation is similar to the weak measurements \cite{TamirCohen2013WeakMeasurements} of the state before undergoing a quantum measurement -- both theory and experiments show that there is a correlation between the weak measurements and the future quantum measurements, but it is not enough to signal back in time \cite{aharonov2012future-past}. However, if by using classical fields one can extract more information about the state while maintaining the disturbance small enough, one could be able to signal back in time. 
Therefore, at least for this case, the most natural way out of this problem seems to be to quantize somehow the gravitational field, so that any measurement we make using gravity disturbs the observed system like quantum measurements do.

Another argument was given by Page and Geilker in \cite{PageGeilker1981IndirectEvidenceQM}, where an experimental result is described, testing whether semi-classical gravity with no wavefunction collapse holds. The authors assume that the unitarity could only be ensured by the Many Worlds Interpretation. The result is interpreted as a falsification of semi-classical gravity because no superposition of macroscopic massive objects at different locations was found. 

In relation to \cite{PageGeilker1981IndirectEvidenceQM}, according to Kent \cite{Kent2018SemiQuantumGravityBellNonlocality}, semi-classical gravity combined with MWI raises the question whether or not psychophysical parallelism (of MWI) ``should or could it apply to the classical gravitational degrees of freedom as well as, or even instead of, the quantum matter?''. He suggests to ``think of the structure of space-time as a (or even the) substrate for a lawlike description of consciousness''. Stating from such questions, Kent analyzes several hybrid models where a single classical spacetime is the same for all branches. Accordingly the choice of the preferred spacetime can be made randomly (according to the Born rule) by one of the Everettian branches, or by some local hidden-variables \etc. 

However, there is an alternative explanation of the results of Page and Geilker in \cite{PageGeilker1981IndirectEvidenceQM}, which favorable to semi-classical gravity without collapse: assume unitary evolution for a single world rather than as in MWI, such that macroscopic massive objects do not end out up in a superposition of being in different positions. According to the proposal supported in this article, global consistency may prevent Schr\"odinger cats, including superpositions of macroscopic massive objects at different locations. Hence, if the hypothesis that global consistency allows for a unitary solution which never collapses discontinuously, and at the same time prevents macroscopic superposition, is true, no serious theoretical or experimental evidence against semi-classical gravity is found in \cite{EppleyHannah1977NecessityQuantizeGravitationalField} and \cite{PageGeilker1981IndirectEvidenceQM}.

Therefore, the absence of a collapse makes Quantum Theory and GR more compatible. But this does not exclude the possibility that quantization of gravity is still needed, and in fact still seems to be required.

Another often-encountered argument that GR has to be changed is that if we would try to experimentally probe spacetime at Planck scales, we would not be able to do it, because the high energies involved would lead to the creation of micro black holes. While indeed this will be the case, this only sets a limitation on our experimental possibilities, one not due to our limited technology, but to in-principle limitations due to backreaction. But the universe is under no contractual obligation to allow us to probe spacetime in these regimes. Such a limitation of our experimental possibilities does not require replacing spacetime with something else, especially since for anything else we would have the same experimental limitations.

Another reason sometimes invoked when semi-classical gravity is said to be only an approximation comes from the \emph{black hole entropy}. It is said that the event horizon is homogeneous, and it would not be able to store information, so either the horizon or the spacetime inside the black hole have to be discrete to represent microstates which would give the right entropy. But the entropy calculations are done with the help of QFT on a spacetime regime. The black hole entropy is calculated based on the quantum information of the particles falling in the black hole \cite{bekenstein1973black,BCH1973fourLawsBH,Wal94,jacobson1996introductory}. Hence, the black hole entropy was derived in the strict framework of QFT on a curved background. The same holds for \emph{black hole evaporation} -- it was derived in the same framework \cite{Haw75ParticleCreation,Haw76}. This means that both black hole entropy and black hole evaporation are predicted and explained by QFT on spacetime already. Then why would we need another theory to derive the same predictions \cite{Stoica2018RevisitingBlackHole}?

Moreover, the information is considered lost only because the black hole singularities seem to lead to Cauchy horizons, but if the global hyperbolicity is not necessarily destroyed by singularities, as explained in \cite{Sto13a}, then this problem can be solved within the framework of GR itself \cite{Sto14c,Stoica2018RevisitingBlackHole}.

All these arguments concur in showing that the main problems we tend to think to require a Quantum Theory of spacetime itself may in fact be solved within QFT on curved spacetime. There is no necessary reason to think otherwise. Surely, it may still be true that there is a better theory, maybe one of the proposals of quantum gravity or maybe one we do not know yet, but the currently known arguments no longer seem to be so strong.

While up to this moment there is no theoretical or experimental challenge to the classical GR spacetime, this may be due to the fact that we are still in an early stage of understanding these phenomena. It could be the case that theoretical progress is made if a no-go theorem is found, which really gives no other choice but to give up classical spacetime. On the experimental side, classical spacetime seems to be safe for now, because of the very high energy required for the experiments we could imagine so far.

But recently a new type of experiment was proposed, which could be more within our present accessibility.
A new proposal was put forward at the same time, independently, by Bose {\etal} in \cite{BoseEtAl2017SpinEntanglementWitnessForQuantumGravity}, and by Marletto and Vedral in \cite{MarlettoVedral2017WhyNeedQuantizeEverything,MarlettoVedral2017GravitationallyInducedEntanglement}. By the name of the main proponents, the suggested experiment is called BMV.
The BMV proposal consists of using the gravitational field to entangle two qubits.
The experiment should be arranged so that the qubits cannot interact with one another directly, but only with the gravitational field. If they turn out to be entangled, it follows that gravity must be quantum, because a classical channel could not mediate the entanglement. A more recent review of this proposal, and replies to the criticism, is contained in \cite{MarlettoVedral2019AnswersQuestionsBMVExperiment}. A recent proposal of a table-top test of quantum gravity by looking for the creation of non-Gaussianity, which be applied to a single rather than multi-partite quantum system, was made in \cite{HowlVedralChristodoulouRovelliNaikIyer2020TestingQuantumGravitySingleQuantumSystem}.

It is very plausible that the BMV experiment or a variation of it can be done in the near future, and that it leads to entanglement mediated by gravitational field.
As the authors explain, a negative result would not exclude the possibility that gravity is quantized.

On the other hand, even a positive result could be explained without necessarily quantizing gravity. A possible way to do this appeals to the mechanism for the EPR experiment proposed in \cite{deBeauregard1977TimeSymmetryEinsteinParadox,Rietdijk1978retroactiveInfluence,price2015disentangling} and discussed in Section \sref{s:no_collapse}. Remember that the retrocausal explanation of the EPR-B experiment suggested in these references is local, and the particles are already separate before being detected by Alice and Bob. Thus, the correlations that are normally attributed to entanglement can be explained by local and separable descriptions of the particles. This approach can also work to transfer entanglement between particles. At the micro level, this retrocausal description looks like a classical channel.
This can be seen as merely a loophole of the BMV experiment, but in a retrocausal theory in which even the EPR-B experiment is explained like this, it is the central explanation, not just a loophole. In other words, one cannot exclude that BMV is explained in a semi-classical way, without appealing to counterfactual definiteness. However, if the same works for both EPR-B and BWV, it would be hard to say that the particles in the EPR-B experiment are quantum, but gravity in the BMV experiment is classical, since they are described similarly.

While I tried to maintain as much as possible a conservative position, I think there are little chances that semi-classical gravity as it is currently understood is the final answer. A reason is historical in nature: in nonrelativistic QM (NRQM), the electromagnetic field was treated through its potential, which was not quantized. Its quantization became necessary only in QFT. How could NRQM work so well while being hybrid in this sense? A possible answer can be given in terms of the representation of the wavefunction on the $3$D-space given in \cite{Sto2019RepresentationOfWavefunctionOn3D}. In that representation, the fields representing individual particles in a product state exist in separate layers of a bundle over the $3$D-space. If they are charged particles, the corresponding electromagnetic potential, which is classical in NRQM, is confined per layer. In other words, each layer of interacting charged particles has its own classical electromagnetic potential. But this means that even NRQM admits, to some extent, superpositions of classical fields.

To see how this applies to the BMV experiment, in the eventuality that the experiment will confirm it, we appeal to a relativistic version of the effect, developed in \cite{ChristodoulouRovelli2019PossibilityQuantumSuperpositionGeometries}. In \cite{ChristodoulouRovelli2019PossibilityQuantumSuperpositionGeometries}, the entanglement between the two particles and the gravitational field is described in the form
\begin{equation}
\label{eq:entanglement_g}
\ket{\psi} = \frac 1 2 \(
\ket{\tn{LL}}\ket{g_{d_{\tn{LL}}}}
+\ket{\tn{RR}}\ket{g_{d_{\tn{RR}}}}
+\ket{\tn{LR}}\ket{g_{d_{\tn{LR}}}}
+\ket{\tn{RL}}\ket{g_{d_{\tn{RL}}}}
\),
\end{equation}
where $\tn{L}$ and $\tn{R}$ indicate whether the particles in the pair are in the left, respectively the right arm of the interferometer, and the states of the form $\ket{g_{d_{\tn{JK}}}}$, $J,K\in\{\tn{L},\tn{R}\}$, correspond to the metric tensor for each term of the superposition.

Recalling the multi-layered representation of the wavefunction on the $3$D-space given in \cite{Sto2019RepresentationOfWavefunctionOn3D}, this means that there is a different metric tensor $\ket{g_{d_{\tn{JK}}}}$ in each of the layers in the superposition. In other words, there is a superposition of geometries. But this only means that each of the terms in the superposition have different metric $\ket{g_{d_{\tn{JK}}}}$ and different covariant derivatives $\nabla_{\ket{g_{d_{\tn{JK}}}}}$, just like in NRQM each term in a superposition of particles may have a different classical electromagnetic potential or connection.

But such superpositions of geometries are not necessarily damaging for spacetime at all. Spacetime is a four-dimensional topological manifold, while the metric tensor $g$ is a field, with all of the geometry it entails. Since in \cite{Sto2019RepresentationOfWavefunctionOn3D} it was shown how all possible entangled states in QM can be represented as fields on the physical space or spacetime, this means that superpositions of geometry can too. 
Spacetime remains a four-dimensional topological and differentiable manifold, but the geometry is generalized so that each layer has its own metric tensor.

In fact, this works even if we quantize gravity and obtain spin $2$ gravitons, since the obtained representation is general enough.

This means that the unitary interpretation and the block universe view proposed here are consistent with Condition \ref{condition:spacetimeg}.

Not only such superpositions of geometries are not against the block universe, but they may help in the perturbative quantization of gravity (Condition \ref{condition:qg}).
It is known that, normally, quantization of gravity cannot be done perturbatively. However, various suggestions of modifications are made, going under the common umbrella of \emph{dimensional reduction} techniques. They are based on \emph{ad-hoc} assumptions which have the purpose to reduce the divergences as the perturbative expansion goes to the UV limit \cite{Sto12d}. It turns out that there is no need to make these \emph{ad-hoc} hypotheses, since the treatment of singularities given in \cite{Sto13a} ensures automatically the conditions imposed in several of these approaches. This means that, when we sum over different classical geometries, in the UV limit and in a position basis, the metric becomes degenerate and cancels the infinities automatically, as required by several of the dimensional reduction proposals \cite{Sto12d}.

%------------------------------------------------------------%
\section{The post-determined block universe}
\label{s:post_determined_block_universe}

As humans, very early in life we become aware that events that already happened cannot be changed, and that future events, although unpredictable, can be influenced by our present actions. This intuition is so deeply hardwired in our world view, that it seems unnatural to even question the idea that past and future do not exist, but only present does. Therefore, it comes as a surprise that thinkers like Parmenides, Boethius, Augustine of Hippo, Anselm of Canterbury, D$\bar{\tn{o}}$gen, and others either claimed, or seem to suggest the opposite view, that past, present, and future are equally real. The first, more common position, is called \emph{presentism}, while the second position is called \emph{eternalism}. A major advocate of eternalism was McTaggart \cite{McTaggart1908UnrealityOfTime}.

With the discovery of classical mechanics,  and its determinism, we started to realize that, at least mathematically, the past, present, and future states of the universe are encoded in the state of the universe at any time $t_0$. As Laplace explained it, if the state of the universe as a physical system is known with perfect accuracy at a time $t_0$, and the equations or motions are known, one can, in principle, calculate the state of the universe at any other time. This is a property of the differential equations expressing the physical laws of classical mechanics. What appear to us as random events are in fact due to the absence of complete information about the state. So it is surprising that, knowing that past and future are encoded in the present, physicists did not consider eternalism as a viable option. We had to await the discovery of Special Relativity, with its relativity of simultaneity, to take this idea seriously. Accordingly, if two observers in relative motion agree that a certain event took place at a certain time, they will in general disagree about what other events happened at the same time with the first one. Spacetime as a block universe (BU) appeared gradually in the works of Lorentz, Poincar\'e, Einstein, and Minkowski \cite{Minkowski1910Spacetime}.
The idea of a spacetime in which time and space are inseparable seemed too farfetched to many physicists, who either rejected Special Relativity, or considered spacetime as a convenient mathematical tool, but eventually the BU and the eternalist position started to be seen as endorsed by Special Relativity.

One of the major objections against the BU was its apparent incompatibility with \emph{free-will} (a notion which I do not know how to define or explain, but it may be more than the one in Condition \ref{condition:free_will}). A hybrid proposal was made by C.D. Broad in 1923, the \emph{growing BU} \cite{broad1923scientificThought}. According to this proposal, both past and present exist, but the future does not yet exist. The past is a block which grows continuously with the passage of time, which is connected to the human experience. 
In a totally different direction, to reconcile free-will with classical determinism, \emph{compatibilists} took the position that freedom means acting according to one's nature, and thus determinism is not only consistent with freedom, but it allows it to be expressed. This compatibilist position is consistent, in particular, with the BU. But Hoefer took a step beyond compatibilism, and proposed that we are even free to make choices affecting the initial conditions of the universe \cite{Hoefer2002Freedom}, within the framework of the BU of classical relativity.

Soon after the discovery of Special Relativity, Quantum Mechanics appeared, with its probabilistic Born rule and Heisenberg's uncertainty principle. This time, it seemed difficult to explain the probabilities as a mere lack of knowledge, and they came to be understood as irreducible. Results like Bell's theorem and the Kochen-Specker theorem \cite{KochenSpecker1967HiddenVariables,Bel66} were taken by many as endorsing this position. This happened despite the fact that at least a deterministic interpretation of QM was known at that time, the {\pwt} of de Broglie and Bohm \cite{Bohm52}, which is consistent with both of these no-go theorems, being both nonlocal and contextual. The presentist position seemed to win.
Because of this, it is usually believed that the BU picture can only hold for classical GR, but quantum indeterminism is necessarily incompatible with it. This is sometimes taken as evidence that the \emph{eternalism} of GR does not hold, and it should be replaced by \emph{presentism}, which can be expected to be consistent with the wavefunction collapse and with nonlocal interpretations of QM.
But the evolving or growing BU model \cite{broad1923scientificThought} was seen by Ellis as consistent with QM, since the growth can be seen as taking place in a nonlocal and indeterministic way as new quantum measurements are performed \cite{ellis2006physicsSpacetimeEBU,ellis2008FlowTimeFQXiEBU}. Later, in 2009, Ellis and Rothman proposed a \emph{crystallizing  BU}, which captures the ``quantum transition from indeterminacy to certainty'' \cite{Ellis2010CrystallizingBlockUniverse}. This model also includes retrocausal influences, and the crystallization does not always occur simultaneously, in some cases having to wait for future experiments to be done, as it is the case with the delayed choice experiment \cite{Whe78}. Thus, the model provides support in particular to retrocausal interpretations of QM in which there is collapse, the collapse leading to the growth of the past block, and quantum indeterminacy allowing the future to be open.

It is possible to conceive a \emph{splitting} or \emph{branching BU}, which may be the BU version of the Many Worlds Interpretation, see for example \cite{Tegmark2014OurMathematicalUniverse}, and for some criticism \cite{ellis2014evolvingBlockUniverse}. One can imagine a branching BU, in which each branch is an evolving or crystallizing BU, but on the overall they are all part of a tree-like structure.
But there is a problem: in MWI, the wavefunction is real, but it is defined on the configuration space. Even for each branch, there is entanglement, at least because each branch contain atoms, which contain entangled particles. So, the wavefunction has to be replaced somehow, with an equivalent structure defined on space or spacetime, rather than on the configuration space. In fact, this applies to the other interpretations which make use of the wavefunction. Fortunately, as I already mentioned, there is a way to replace the wavefunction by a high-dimensional field defined on space, which is mathematically equivalent to the wavefunction on configuration space \cite{Sto2019RepresentationOfWavefunctionOn3D}.

Despite the tremendous success of nonrelativistic Quantum Mechanics in describing microscopic physical phenomena, the domain of high energies requires something more. Quantum Field Theory appeared as a relativistic quantum theory. Wigner and Bargmann classified elementary particles in terms of representations of the Poincar\'e group \cite{Wig39,bgm:1954}. Thus, as long as there is no collapse, QFT is consistent with relativistic invariance, and it even requires it at the deepest level. But the wavefunction collapse would be in tension with this relativistic invariance, required by the very definition of particles in terms of representations of the Poincar\'e group. This is another reason to take unitary evolution seriously, and this brings us to the central point of this article.

In this article, I argued that in Quantum Mechanics it is possible for the time evolution to be unitary, without discontinuous collapse, at the level of a single world. I argued that this has some advantages, in particular it allows us to have a single law which is not broken even occasionally, it ensures the conservation laws, avoids nonlocality, and makes QFT more consistent with GR, but seems to suggest that the initial conditions have to be very specially chosen, to conspire to make measurements have definite outcomes rather than resulting in Schr\"odinger cats. But these apparent conspiracies seem rather natural if we see them as the consequences of global consistency in a BU. I hypothesized that the global consistency is about the consistency of solutions in the presence of topological constraints of the fiber bundles from spin geometry and gauge theory, and that it manifests itself by reducing the Hilbert space to a subset which does not contain Schr\"odinger cats and where all quantum measurements have definite outcomes.

These arguments suggest the following picture of a \emph{post-determined block universe}. There is a BU, on which QFT is true, and when formulated in the multi-layered representation \cite{Sto2019RepresentationOfWavefunctionOn3D}, in each of the geometries taking part in the superposition, the mean value of the stress-energy operator connect to the spacetime curvature via the Einstein equation \eqref{eq:einstein_eq}.
Alternatively, if semi-classical gravity turns out to be valid despite the general position that it is not enough, the mean value of the overall stress-energy tensor is plugged into the Einstein equation \eqref{eq:einstein_eq}.

Yet not all possible initial conditions can be true, but only those which lead to globally consistent solutions. So we start with a set of possible BUs consistent with the observed coarse graining state, including the previous quantum observations, and as we make new observations, we refine that set by eliminating all of the solutions that are not consistent with the new observations. So far it looks quite like classical physics, but since the initial conditions are severely constrained by global consistency, the resulting correlations can be expected to violate the Bell inequality and its generalizations in a way in which a dynamical system with no global constraints of the initial conditions could not be able to violate. This allows us to have a picture in which quantum measurements happen as predicted by the projection postulate, without breaking the unitary time evolution. In this picture, the BU is not pre-determined, but it is gradually post-determined as we make new quantum observations.

This kind of BU is deterministic, but it is not predetermined in the usual sense. The initial conditions are determined with a delay, by each new measurement and each choice of what to measure. The requirement of global consistency implies a severe restriction on the possible solutions of the Schr\"odinger equation, but since the observers can choose what to measure, it looks like they determine the past initial conditions more, with each new choice. The solution is still deterministic, but it is determined by future choices and the outcomes of measurements. 
We can still think of this proposal as including a form of superdeterminism or retrocausality, if we assume that the initial conditions are fixed from the beginning. But we can also take the position that the quasi-classical limit, which is a coarse graining of the low-level quantum state, evolves by usual weak causality in an indeterministic way (\cf Condition \ref{condition:weak_causality}). As observers, we start with the full set of quantum states consistent with our previous macroscopic observations, and then reduce them as new measurements provide more information. And since we never know the true quantum state, but only outcomes of our observations made on subsystems, these observations allow us to predict only probabilities, or an epistemic wavefunction which is an approximation of the ontic wavefunction, and has to be readjusted after each new measurement by invoking the wavefunction collapse.

For the reader concerned about the determinism inherent in the post-determined BU, there is freedom in the initial conditions. The proposal is consistent with ``free-will'' both in the sense of Condition \ref{condition:free_will}, and like in Hoefer's model \cite{Hoefer2002Freedom}, but, in addition, it works in the context of QM, to explain measurements by unitary evolution alone. In relation to unitary evolution without collapse, this type of BU was suggested in various places by the author \cite{Sto08b,Sto08e,Sto08f,Sto12QMc,Sto13bSpringer}.
The post-determined BU can be compared with the crystallizing BU proposal \cite{Ellis2010CrystallizingBlockUniverse}, since the latter also has retrocausal influences due to delayed choices, but in the post-determined model, these apparent retrocausal influences apply to the entire past history of the universe. The post-determined BU can also be compared with the splitting BU of MWI, but there is no splitting at the fundamental level, only at the coarse grained level we can consider that there is branching. We start with a set of possible block universes consistent with our current observations, and as we add new observations, we refine the set of possible block universes, by eliminating those that are inconsistent with the new data. Indeterminism, which can also be interpreted as branching, is manifest only at the macroscopic or coarse grained level. In the post-determined BU, here is no actual growth or crystallization or branching at the fundamental level, but only at the coarse grained level accessible to the observers.
The post-determined BU accommodates the main advantage of these proposals, which is their consistency with Quantum Mechanics, by solving the problems mentioned in \sref{s:collapse_problems}, and restores the full advantages of the BU picture at the same time. 

The post-determined BU is as deterministic and fixed as the standard one from the bird's eye view of someone who knows completely the ontic wavefunction of the universe. From the point of view of someone who is part of the universe itself, like us, it may look as a growing BU, with the amendment that the growth is not only towards the future, but at the quantum scale, because of global consistency, it also seems to be growing towards the past, giving the impression of retrocausality. But this retrocausality is not accessible to us to send messages into the past or at a distance, being forbidden by the fact that we only have ``clearance'' to approximate eigenstates, and not to the full quantum state of the observed systems.

By eliminating the discontinuous collapse, we remove important obstructions that seemed to put Quantum Theory and General Relativity at odds with each other, as seen in Section \sref{s:quantum_and_relativity}.

%------------------------------------------------------------%
\section{Empirical predictions of the post-deterministic interpretation}
\label{s:experiment}

While by the name ``interpretation of Quantum Mechanics'' one expects that such theories are indistinguishable from standard QM, the type of unitary interpretation discussed here makes different empirical predictions. The entire quantum part of the proposed theory is deduced from the assumption that the evolution happens for our single world all the time according to the {\schrod} equation \eqref{eq:unitary_evolution}, and there is no collapse, for example through Theorems \ref{thm:collapse_breaks_conservation} and \ref{thm:unitary_ic}. 

Since the presence of a collapse implies violations of conservation laws \cf Theorem \ref{thm:collapse_breaks_conservation}, if we would be able to look for such violations, we would be able to distinguish unitary models at the level of a single world, of the type discussed here, from interpretations in which collapse is genuine (for example the Copenhagen Interpretation and GRW) or effective collapse for single worlds which is absent for the universal wavefunction (as in MWI and {\pwt}).

An experimental verification of the conservation laws during measurements involves the measurement of the total system, which consists of the observed system and the apparatus (including the relevant environment). Suppose we want to search for violations of the angular momentum of a Silver atom during a quantum measurement. Then the difference between the angular momentum of the observed atom before and after the measurement should transfer somehow (\ie through local interaction) to the apparatus. However, several problems appear here: the apparatus is very large by comparison to the observed system and the difference of angular momentum we are looking for is very small, to measure the total angular momentum of the combined system before and after the measurement is tantamount to a kind of \emph{Wigner's friend experiment} \cite{Wigner1962WignerFriend}, the reading of the display of the apparatus by the observer would induce additional changes, the environment should be isolated completely, and quantum fluctuations and Heisenberg uncertainty may perturb significantly such a small quantity. All these make very impractical this type of experiment, at least if we are looking to measure directly the involved quantities for the combined system.

A more concrete idea was proposed by Schulman \cite{schulman2012experimentalTestForSpecialState,schulman2016specialStatesDemandForceObserver,schulman2016lookingSourceChange}, based on observing a magnetic field interacting with a Silver atom in a modified Stern-Gerlach experiment. If the angular momentum of the atom changes into an eigenstate, if there would be a collapse, this change would not require an interaction in order to reorient the atom's angular momentum. But if the evolution is unitary, this change has to be realized by an interaction. Due to the apparently conspirational nature of such interactions, the interaction can take place during various phases of the experiment, but Schulman estimates that the most likely interaction is due to a fluctuation of the magnetic field of the Stern-Gerlach device itself, which is not required by other interpretations where collapse takes place or is at least effective for each branch. The experiment consists of two parts, both involving Stern-Gerlach devices: a preparation of the atom in a certain eigenstate, and a detection. If the evolution of the total system is unitary without collapse, it is expected that the magnetic field gives a kick to the angular momentum of the atom, to reorient it so that it is found in an eigenstate. If the rotation needed is by an angle $\theta$, it can be accomplished as well by any angle $\theta+2n\pi$, for any integer $n$. For large values of $\abs{n}$, the fluctuation of the magnetic field is larger, and therefore easier to detect. In this case, the fluctuation of the field can be large enough to compensate for the Heisenberg uncertainty. 

As another way to make the effect more detectable, Schulman proposes to prepare ensembles of atoms and send them together through the second Stern-Gerlach device. For ensembles, it is possible that a single larger fluctuation of the magnetic field reorients more atoms at once, which would lead to an observation of clusters of consecutive outcomes corresponding to the same eigenvalue, alternated with clusters of outcomes corresponding to the other eigenvalue, so that the average is as predicted by Quantum Mechanics. There is no reason to expect such clustering if the collapse is genuine, but if an interaction is needed to reorient the spin, it is more likely that the same fluctuation affects more atoms in the ensemble. Schulman estimates that both these experimental tests are feasible with the current capabilities.

Another possibility arises from Schulman's proposal, which is due to the fact that the model is not known to predict the Born rule. In order to obtain the Born rule, Schulman attempts to derive the distribution of fluctuations as a function of the rotation angle $\theta$ capable to give the right probabilities. However, he finds out that there is no such normalizable function, so there is no probability distribution which would give the precise Born rule (\cite{schulman2016specialStatesDemandForceObserver} Appendix B). The problem would be solved if the function would be $f_0(\theta)=\frac 1{\theta^2}$, but this is not normalizable. But by making use of the Cauchy distribution $C_a(\phi)=\frac{a/\pi}{a^2+\phi^2}$, Schulman shows that there is a probability distribution that can be made close enough to give the Born rule with any precision, by appropriately choosing the parameter $a$. Schulman proposes that the parameter $a$ can be bounded from experiments. But it is possible that, since there is no actual probability distribution which would give the exact Born rule without requiring a renormalization, the parameter $a$ is in reality fixed and leads to observable deviations from the Born rule. Therefore, while it is not clear that Schulman's proposal violates the Born rule, since his distribution is able to approximate it with any desired precision, observing violations of the Born rule may fix the parameter $a$, which would allow to test if indeed the probability distribution is the one proposed by Schulman. It seems very unlikely to find systematic deviations from the Born rule, since it has unique properties, as known for example from Gleason's theorem \cite{GleasonTheorem1957}. However, since in the unitary single world interpretations of this type measurements are very special events compared to the events allowed by quantum superposition, by requiring special configurations of the observed system and the measurement apparatus, it is not excluded that such deviations from the Born rule exist for very small values of the angle $\theta$.

In \cite{schulman2016lookingSourceChange}, Schulman proposes another experiment, also based on angular momentum, but this time for photons. The experiment consists in sending a circularly polarized photon through a birefringent crystal, then through a beam splitter which polarizes it linearly, then measuring the polarization, and detecting unusual movement taking place in the birefringent crystal, as a result of the change in angular momentum due to the detection. Again, a theory or interpretation of quantum mechanics in which collapse or effective collapse takes place, no interaction is needed to change the polarization of the photon. But if everything happens unitarily, interactions should account for the change in polarization, and the birefringent crystal should change its motion. Schulman's experiment starts by preparing the photon in a state of helicity $1$ (angular momentum $\hbar$) along the $z$-axis, by sending it through the birefringent crystal. Then, the photon encounters a beam splitter that separates it into $S$ and $P$ polarizations -- linear polarizations along the $x$ and $y$-axes respectively, both with $0$ helicity along the $z$-axis. Then, the photon is detected, which allows to infer its polarization state. Schulman estimates that the change of the angular momentum is most likely due to the interaction with the birefringent crystal. Of course, the crystal has to be small enough to change its state of motion accordingly in an observable way, but Schulman is optimistic about the feasibility of this experiment, due to the current possibilities to make microscopic birefringent crystals \cite{AritaMaziluDholakia2013LaserInducedRotationAndCoolingOfTrappedMicrogyroscopeInVacuum} and measuring small changes in angular momentum \cite{RothmayerEtAl2009IrregularSpinAngularMomentumTransfer,KolesovEtAl2013MappingSpinCoherence}.
However, I think that this experiment should be designed and realized while taking into account that, if the birefringent crystal is very small, by interacting with the photon the two may become entangled, and standard Quantum Mechanics may predict the same correlation between the observed polarization and the motion of the crystal, in a similar way as in the EPR experiment. If this possibility is not excluded from the experiment, this would make it useless for the present purpose.

To summarize, there are already proposals that are expected to distinguish empirically between theories which involve wavefunction collapse on the one hand, whether it is a genuine collapse as in the Copenhagen Interpretation or collapse theories like GRW, or merely effective collapse as in MWI and {\pwt}, and on the other hand theories with unitary evolution and no collapse at the level of single worlds like the ones discussed here.

%------------------------------------------------------------%
\section{Open problems}
\label{s:open}

The full construction presented in this article was based on the unitary interpretation for a single world \cite{schulman1997timeArrowsAndQuantumMeasurement,Sto08b,Sto12QMb,Sto12QMc,Sto13bSpringer,Sto16aWavefunctionCollapse,Sto17UniverseNoCollapse}, on the representation of the wavefunction on the physical space \cite{Sto2019RepresentationOfWavefunctionOn3D}, on the theory of singularities developed in \cite{Sto13a}, on its application to quantum gravity in \cite{Sto12d} and to the black hole information paradox \cite{Sto14c,Stoica2018RevisitingBlackHole}, and on the sheaf theoretical framework for physics initiated in \cite{Sto08WorldTheory}.
As a result, the model is consistent with the Conditions 
\ref{condition:ontology},
\ref{condition:spacetimeQ},
\ref{condition:spacetimeg},
\ref{condition:universal_law},
\ref{condition:conservation},
\ref{condition:relativity_of_simultaneity},
\ref{condition:locality},
\ref{condition:strong_causality},
\ref{condition:weak_causality},
\ref{condition:free_will},
\ref{condition:quasiclassicality},
\ref{condition:probabilities},
\ref{condition:qeinstein},
\ref{condition:qg},
\ref{condition:singularities}, and
\ref{condition:bh-info} mentioned in Section \sref{s:intro}.
However, despite being consistent with Condition \ref{condition:probabilities}, it is not known yet if it predicts the Born rule, or if more details or constraints are needed to derive it. This is an open problem at this time. In fact, it is a decisive test for the theory, since if the theory predicts a different probability distribution than what we obtain in the experiments, it should be considered falsified.

Another problem is that, in the absence of a sheaf-theoretic derivation of the theory, I proposed an \emph{ad-hoc} effective model, by postulating Principle \ref{principle:constraints} in Section \sref{s:global_consistency}. Hopefully, a more fundamental explanation will be found in the future.

A possible way to do this may emerge from the study of the sheaf cohomology of the fiber bundles used in gauge theory, as mentioned as a possibility in Section \sref{s:global_consistency}. This remains an open research program. However, experiment is called to decide on its fate, and some proposals were discussed in Section \sref{s:experiment}.

\textbf{Acknowledgment}

The author wishes to thank the anonymous reviewers whose suggestions helped improving the quality of the article.

On behalf of all authors, the corresponding author states that there is no conflict of interest. 

%-----------------------------------------------------%

\end{document}